\documentclass[pra, twocolumn, showpacs,superscriptaddress,notitlepage]{revtex4-1}
\usepackage[english]{babel}

\usepackage{amsmath,amsthm,amssymb}
\usepackage{bbold}
\usepackage{graphicx}
\usepackage{bm}
\usepackage{braket}
\usepackage{physics}
\usepackage{subcaption}
\usepackage{enumitem}
\usepackage{algorithm}
\usepackage{algpseudocode}

\DeclareMathOperator*{\argmin}{argmin}

\usepackage{mathtools} 
\DeclarePairedDelimiter{\altket}{\vert}{\rangle}

\algnewcommand\algorithmicinput{\textbf{Input:}}
\algnewcommand\Input{\item[\algorithmicinput]}
\algnewcommand\algorithmicoutput{\textbf{Output:}}
\algnewcommand\Output{\item[\algorithmicoutput]}
\algnewcommand\algorithmicruntime{\textbf{Runtime:}}
\algnewcommand\Runtime{\item[\algorithmicruntime]}

\newcommand{\expt}[3]{\left\langle #1 \middle| #2 \middle| #3 \right\rangle}
\newcommand{\exptshort}[3]{\bra{#1} #2 \ket{#3}}

\theoremstyle{plain}
\newtheorem*{theorem*}{Theorem}

\makeatletter
\newcommand{\labeltext}[3][]{%
    \@bsphack%
    \csname phantomsection\endcsname
    \def\tst{#1}%
    \def\labelmarkup{\emph}
    \def\refmarkup{}%
    \ifx\tst\empty\def\@currentlabel{\refmarkup{#2}}{\label{#3}}%
    \else\def\@currentlabel{\refmarkup{#1}}{\label{#3}}\fi%
    \@esphack%
    \labelmarkup{#2}
}
\makeatother

\usepackage{qcircuit}    

\usepackage[colorlinks=true, allcolors=blue]{hyperref}

\begin{document}

\title{Scalable Quantum Computation of Highly Excited Eigenstates with Spectral Transforms}
\author{Shao-Hen Chiew}
\affiliation{Department of Physics, Faculty of Science
National University of Singapore
Blk S12 Level 2, Science Drive 3
Singapore 117551}
\affiliation{Centre for Quantum Technologies, National University of Singapore, 3 Science Drive 2, Singapore 117543}
\author{Leong-Chuan Kwek}
\affiliation{Centre for Quantum Technologies, National University of Singapore, 3 Science Drive 2, Singapore 117543}
\affiliation{MajuLab, CNRS-UNS-NUS-NTU International Joint Research Unit, Singapore UMI 3654, Singapore}
\affiliation{National Institute of Education, Nanyang Technological University, Singapore 637616, Singapore}
\affiliation{Quantum Science and Engineering Center (QSec), Nanyang Technological University, Singapore}
\begin{abstract}
	We propose a natural application of Quantum Linear Systems Problem (QLSP) solvers such as the HHL algorithm to efficiently prepare highly excited interior eigenstates of physical Hamiltonians in a variational and targeted manner. This is enabled by the efficient computation of the expectation values of inverse Hamiltonians on quantum computers, in situations where Hamiltonian simulation and the representation of eigenstates on quantum computers are efficient. Importantly, the usage of the QLSP solver as a subroutine within our algorithm -- with its inputs and outputs corresponding to physically meaningful objects such as Hamiltonians and eigenstates arising from physical systems -- does not conceal exponentially costly pre/post-processing steps that usually accompanies it in generic linear algebraic applications. We detail implementations of this scheme for both fault-tolerant and near-term quantum computers, analyze their efficiency and implementability, and detail conditions under which the QLSP solvers' exponentially better scaling in problem size render it advantageous over existing classical and quantum approaches. Simulation results for applications in many-body physics and quantum chemistry further demonstrate its effectiveness and scalability over existing approaches.
\end{abstract}

\maketitle

\section{Introduction}
The study of the spectral properties of physical systems plays a central role in physics, chemistry, and materials, and constitutes one of the leading candidates for quantum computers to display an advantage over classical approaches \cite{bauer2020quantum,mcclean2016theory,preskill2018quantum,cerezo2021variational,bharti2021noisy,eisert2020quantum}. While the search for practical problems that definitively satisfy such a claim remains an ongoing work \cite{babbush2021focus,lee2023evaluating}, the identification of problems and algorithms where quantum computers have the potential to yield an advantage remains a highly pertinent issue.

Eigenstate preparation is an important task within this category of problems. For the case of ground states, the problem is generically difficult for a large class of local Hamiltonians even for quantum computers \cite{kempe2006complexity}, but the presence of additional structure -- such as knowledge of the ground state energy and gap, and the availability of initial quantum states with reasonable overlaps with the ground state -- renders the problem tractable, and gives way to the design of more efficient algorithms, which remains an active area of research \cite{ge2019faster,lin2020near}. Variational quantum eigensolvers (VQE) constitute another interesting approach that are consistent with near-term NISQ devices due to their low memory and depth requirements \cite{peruzzo2014variational}. While these algorithms lack vigorous runtime guarantees due to their variational nature, they can supply quantum states with high overlap with the ground state at low cost, 
effectively acting as subroutines of more powerful algorithms described above.

On the other hand, the case for highly excited eigenstates has received less attention, and is generally more challenging than the ground state problem \cite{gonzalez2012progress,westermayr2020machine}. Nonetheless, the development of theoretical and computational methods to study excited states remain crucial for a plethora of reasons, from the understanding of the ergodic-localization phase transition of disordered many-body systems \cite{bauer2014analyzing,smith2016many,pietracaprina2018shift}, the calculation of reaction rates and binding energies of molecules \cite{wang1994solving,reiher2017elucidating,cao2019quantum}, vibrational spectroscopy \cite{baiardi2019optimization} and photochemistry \cite{lindh2020quantum}, to an understanding of biological processes such as photosynthesis \cite{cerullo2002photosynthetic} and human vision \cite{herbst2002femtosecond}. Classically, when analytical or approximative treatments are inadequate, one resorts to adaptations of ground state methods such as Lanczos methods combined with spectral transforms \cite{pietracaprina2018shift} and machine learning methods \cite{westermayr2020machine}, but they are ultimately hampered by demanding memory and time requirements that generically scale exponentially with system size, especially when considering interior eigenstates located deep within the bulk of the spectrum.

This prompted the question of whether quantum algorithms can fare better. Along this direction, a quantum algorithm adopting a filtering approach was recently introduced, which requires similar structures as ground state methods described above, and high-overlap initial states obtained through approaches such as those based on adiabatic quantum computing \cite{lin2020optimal}. On near-term quantum computers, despite recent developments \cite{peruzzo2014variational,mcclean2017hybrid,santagati2018witnessing,parrish2019quantum,higgott_variational_2019,nakanishi2019subspace,zhang2021adaptive,mondal2023ground}, most approaches can only explore a limited number of eigenstates close to the ground state, lacking scalability for highly excited eigenstates. This mirrors the limitations of classical approaches based on iteratively orthogonalizing against lower-energy eigenstates, which becomes prohibitively expensive when targeting highly excited interior eigenstates of large systems \cite{dorando2007targeted,baiardi2019optimization}.

The application of quantum computers to linear algebraic tasks constitutes another promising avenue with wide-reaching applications. Among them is the solution of linear systems of equations and its quantum variant, the Quantum Linear Systems Problem (QLSP), which, given a $N \times N$ dimensional matrix $A$ and a $N$ dimensional vector $\vec{b}$ seeks the preparation of a quantum state $\ket{x} = \sum_{i=1}^N x_i \ket{i} / \lVert \sum_{i=1}^N x_i \ket{i} \rVert $ that encodes the solution $\vec{x} = (x_1,...,x_N)^\top$ of the linear system $A\vec{x} = \vec{b}$. An efficient solution for the QLSP on quantum computers was first proposed by Harrow,
Hassidim and Lloyd (HHL) \cite{harrow2009quantum}, which yielded an exponential improvement in the scaling in $N$ over known general classical algorithms. Since then, a flurry of work have resulted in improved algorithms with better scalings \cite{berry2015hamiltonian,childs2017quantum,low2017optimal,gilyen2019quantum,an2022quantum,costa2022optimal}. QLSP solvers that are feasible in near-term devices have also been proposed and tested on current quantum computers \cite{subacsi2019quantum,bravo2019variational,xu2021variational,huang2021near}. 

However, the exponential improvement over classical solvers comes with several important caveats that prevent straightforward application \cite{aaronson2015read}. Besides requiring $A$ to be sparse, well-conditioned, and efficiently simulable, QLSP solvers output a quantum state, only allowing efficient access to its statistical properties such as expectation values. Access to wavefunction amplitudes require further processing, which may scale as the exponentially many number of amplitudes if complete information is desired. The encoding of the entries of $\vec{b}$ as amplitudes of a quantum state $\ket{b}$ is also non-trivial, facing the same complexity for unstructured inputs. Proposed applications of QLSP solvers must therefore conform to these potentially limiting requirements to truly harness the exponential improvement, of which there are limited instances \cite{wiebe2012quantum,clader2013preconditioned,chowdhury2016quantum,wang2017efficient}. Indeed, the search for such applications constitutes an interesting and difficult problem in itself.

In this work, we propose and study an algorithm which prepares highly excited eigenstates of physical and chemical Hamiltonians on both near-term and fault-tolerant error-corrected quantum computers. Inspired by classical targeted eigenvector algorithms based on spectral transformations, it prepares arbitrary excited eigenstates near a specified target energy in a variational manner, which is enabled by the computation and optimization of the expectation value of a shift-inverted Hamiltonian via the usage of QLSP solvers. Besides the algorithm's pragmatic value to the eigenstate preparation problem, it also presents a highly natural context for the application of QLSP solvers. By taking $A$ and $\ket{b}$ to be Hamiltonians and quantum states arising from physical contexts and only requiring the extraction of statistical properties from the solution $\ket{x}$, we find that exponentially costly processing steps that usually accompany QLSP solvers in generic linear algebraic applications are avoided.

Furthermore, we evaluate conditions and settings under which the proposed algorithm can outperform existing methods (both classical and quantum), which depends primarily on whether the representation of target eigenstates and Hamiltonian simulation can be performed more efficiently on quantum computers. This leads us to the conclusion that problems arising from physical contexts such as chemistry and many-body physics constitute promising candidates that enable our proposed algorithm to outperform classical implementations of a similar procedure. Numerical simulations of our algorithm applied to problems in the context of quantum chemistry and disordered many-body localized (MBL) systems further supplement our arguments, with results indicating superior performance over existing methods for the same task. We also extensively discuss our algorithm's relation and advantages over other existing quantum algorithms, most of which require exponentially many circuit executions with increasing system size when targeting highly excited eigenstates deep in the spectrum. Instead, our approach prepares eigenstates in a targeted manner, with an efficiency that depends mainly on the local density of states near the target, rendering it particularly favorable for applications in the near-term.

The article is structured as follows. Section.~\ref{sec:algorithm} acts as the central section of the work, which provides a general description of the proposed algorithm and task. It highlights the algorithm's main features and components in an implementation-independent manner, while referencing all other sections of the text. Focusing on the core QLSP solver subroutine, Section.~\ref{sec:implementations} details the fault-tolerant and near-term implementations of the algorithm (at Sections.~\ref{sec:fault_tol} and \ref{sec:near_term} respectively), followed by analyses of their complexities and implementational costs (at Sections.~\ref{sec:fault_tol_anal} and \ref{sec:near_term_analysis} respectively). Section.~\ref{sec:other_components} details other components beyond the QLSP solver, which completes the description and analysis of our proposed algorithm. Next, Section.~\ref{sec:advantage} provides a detailed discussed on situations and conditions under which the algorithm can outperform existing classical implementations, while Section.~\ref{sec:relation_quantum} describes advantages and synergies with existing quantum algorithms. Finally, Section.~\ref{sec:applications} describes results of numerical simulations of our algorithm applied to problems arising from quantum chemistry (Section.~\ref{sec:applications_chem}) and many-body localized spin systems (Section.~\ref{sec:applications_spin}). We conclude with a discussion on the implications of our proposed algorithm, future avenues of research, and possible limitations in Section.~\ref{sec:discussion}. We also refer the reader to the Appendices for details, where we review the HHL algorithm (Appendix.~\ref{appendix:hhl}), provide results of additional numerical simulations (Appendix.~\ref{appendix:additional_numerics}), and elaborate on other applications of our results (Appendix.~\ref{appendix:applications}).

\section{The Algorithm} \label{sec:algorithm}

This section outlines the task and the proposed algorithm in a general and implementation-independent manner, which is followed by brief descriptions of its key features and components. References to other sections throughout the text provide further elaboration.

Let $H$ be an $n$-qubit Hamiltonian of dimension $N \times N$ (so that $N = 2^n$), condition number $\kappa$, and sparsity $s$, with eigenstates $\ket{\lambda_i}$ with energies $\lambda_i$. We consider the task of preparing the $k$-th excited eigenstate $\ket{\lambda_k}$ on a quantum computer, with the additional structure that an estimate of $\lambda_k$, which is separated from its closest neighboring eigenvalue by a spectral gap $\Delta$, is known. Denote $\tilde{H} = H-\sigma\mathbb{1}$ the Hamiltonian $H$ shifted by a real constant $\sigma\mathbb{1}$, and $\ket{0...0}$ a quantum state that is easily preparable on a quantum computer.

While the procedure to be described can be applied to any input matrix, we focus on $k$-local Hamiltonians arising from physical contexts, which can be expressed as linear combinations of tensor product of Pauli matrices:
\begin{equation} \label{eq:k-local}
    H = \sum_{i=1}^L c_i h_i = \sum_i c_i \bigotimes_{j=1}^n \sigma_{j}^i,
\end{equation}
where each $\sigma_{j}^i \in \{\mathbb{1}, \sigma^x, \sigma^y, \sigma^z \}$, the total number of terms $L$ is at most polynomial in $n$, and each of the $L$ Pauli terms $h_i = \bigotimes_{j=1}^N \sigma_{j}^i$ in the sum acts non-trivially on at most $k$ qubits. For instance, spin models describing magnetism appear naturally in this form, usually with additional geometrical structure and locality, while fermionic Hamiltonians of molecular systems in chemistry can be cast into this form, e.g. by applying a Jordan-Wigner transformation to a second-quantized electronic Hamiltonian \cite{whitfield2011simulation}. Problems considered in our numerics in Section \ref{sec:applications} also take this form.

The description of the algorithm is then as follows:
\begin{enumerate}
    \item \textbf{Shift selection} : Choose a value $r \in (0,1]$. If the neighbor's energy is higher than $\lambda_k$, set the shift $\sigma = \lambda_k + r \Delta$; otherwise set $\sigma = \lambda_k - r \Delta$.

    \item \textbf{Variational state preparation} : Prepare a parametrized quantum state:
\begin{equation} \label{eq:input_state}
    \ket{b(\theta)} = U(\theta) \ket{0}^{\otimes n} = \sum_i \beta_i \ket{\lambda_i}
\end{equation} 
with a unitary ansatz $U(\theta)$, where $\ket{0}^{\otimes n}$ is an easily preparable quantum state. $U(\theta)$ is chosen to be able to efficiently represent the target eigenstate $\ket{\lambda_k}$, which depends on the form of $H$.
    
    \item \textbf{Inverse expectation estimation} : Compute an approximation to:
    \begin{equation} \label{eq:cost_func}
        C(\theta) = \bra{b(\theta)} \tilde{H}^{-1} \ket{b(\theta)}.
    \end{equation}
    \begin{enumerate}
        \item \textbf{Solution of QLSP} : Obtain the quantum state $\ket{\tilde{x}} \approx \tilde{H}^{-1}\ket{b(\theta)}/ \lVert \tilde{H}^{-1}\ket{b(\theta)} \rVert$ by solving the QLSP with inputs $\tilde{H}$ and $\ket{b(\theta)}$ with a quantum computer up to an error $\epsilon_{\text{HHL}} = O(\Delta)$.
        
        \item \textbf{Expectation estimation} : Compute the expectation value for the observable $H-\sigma\mathbb{1}$ on the output state $\ket{\tilde{x}}$ up to an error $\epsilon_{\text{exp}} = O(\Delta)$, which approximates $C(\theta)$ (up to a normalization constant $\lVert \tilde{H}^{-1}\ket{b(\theta)} \rVert$).
    \end{enumerate}
    
        \item \textbf{Classical optimizer feedback} : With the aid of a classical optimization algorithm, solve the optimization problem:
        \begin{equation} \label{eq:opt_prob}
            \theta^* = \argmin_{\theta} \bra{b(\theta)} \tilde{H}^{-1} \ket{b(\theta)},
        \end{equation}
        i.e. search for optimal parameters $\theta^*$ that optimizes $C(\theta)$ (minimized if $\sigma > \lambda_k$, maximized if $\sigma < \lambda_k$), which consists of repeated executions of Steps 2 and 3. Upon convergence, the optimal parameters $\theta^*$ are stored in classical memory, which allows an approximation of $\ket{\lambda_k}$ to be prepared readily via:
        \begin{equation} \label{eq:solution_prep}
            \ket{\lambda_k} \approx \ket{b(\theta^*)} = U(\theta^*) \ket{0}^{\otimes n}.
        \end{equation}
\end{enumerate}

The core motivation behind the above procedure is the observation that quantum computers enable both the computation of the cost function Eq.~(\ref{eq:cost_func}) and the representation of eigenstates $\ket{\lambda_k}$ to be achieved in a natural, and possibly efficient manner. This allows the ground state problem of a shift-inverted Hamiltonian $\tilde{H}^{-1}$ to be solved -- essentially by applying the variational principle to $\tilde{H}^{-1}$ to solve the optimization problem Eq.~(\ref{eq:opt_prob}) --  which yields an excited eigenstate of the original Hamiltonian $H$, due to the form of the shift-inversion. The output is then a set of parameters $\theta^*$ that can be stored in classical memory, to readily prepare an approximation of the target eigenstate on a quantum computer via Eq.~(\ref{eq:solution_prep}). More generally, the transformation $H \rightarrow (H-\sigma)^{-1}$ is an instance of a \textit{spectral transformation} $H \rightarrow f(H)$, where the function $f$ is chosen to re-order the spectrum of $H$ without changing its eigenstates. It is further chosen such that eigenstates of $H$ with energies closest to $\sigma$ become exterior eigenstates of $f(H)$, so that the solution of the ground state problem of $f(H)$ yields an interior eigenstate of $H$.

Importantly, the usage of the QLSP solver in our context circumvents its usual caveats that are present in generic linear-algebraic applications \cite{aaronson2015read}. Firstly, no expensive pre-processing to map/load the input quantum state $\ket{b(\theta)}$ from a classical vector $\vec{b}$ is involved. Moreover, the QLSP's input $\ket{b(\theta)}$ possess structure, in that it should ultimately correspond to target eigenstates of physical systems, which implies that $U(\theta)$ can be chosen in a physically motivated manner with possibly efficient implementations such as adaptive \cite{grimsley2019adaptive,tang2021qubit} or unitary coupled-cluster type \cite{lee2018generalized,greene2021generalized} ansatzes for chemistry and the Hamiltonian Variational Ansatz for lattice models \cite{wecker2015progress}. Similarly, only easily accessible statistical quantities are extracted from the output quantum state $\ket{x}$, via measurement of the observable $H$ to compute the cost function Eq.~(\ref{eq:cost_func}); no mapping to a vector $\vec{x}$ is involved. Finally, the input to the QLSP solver corresponds to Hamiltonians of the form Eq.~(\ref{eq:k-local}) arising from physical contexts, which are sparse and efficiently simulable (or more generally, they possess efficient input models). For these reasons, a highly natural context for the application of QLSP solvers is presented.

Instead of targeting an eigenstate with a known energy, an alternative usage of the algorithm is to target an eigenstate with energy closest to a shift $\sigma$, where $\sigma$ is selected in physically interesting regions. Depending on the context, strategies to systematically select or guess the shift are known, which we describe in Sections.~\ref{sec:applications_chem} and \ref{sec:applications_spin}. We show how the algorithm can be used in practice for both cases, to (i) directly prepare a target eigenstate when an estimate of its energy is known (see Section.~\ref{sec:applications_chem} for numerical example for a quantum chemistry problem), or to (ii) piece together information on an entire excitation spectrum in a coarse-grained manner, by successively targeting eigenstates at different energy densities without a priori knowledge of their energies (see Section.~\ref{sec:applications_spin} for example to obtain information on the spectrum of an MBL Hamiltonian). Physical properties can subsequently be extracted from prepared eigenstates via efficient measurements of global statistical properties such as expectation values corresponding to dipole moments, magnetization, charge densities and so on, and via quantum quenches for dynamical characterization.

A key step of the procedure is Step 3, the computation of the expectation value of a shift-inverted Hermitian operator Eq.~(\ref{eq:cost_func}) on a quantum computer. This is enabled by the solution of the QLSP with inputs $\tilde{H}$ and $\ket{b(\theta)}$ (Step 3.(a)) which yields:
\begin{equation} \label{eq:qlsp}
    \ket{\tilde{x}} \approx \frac{\tilde{H}^{-1}\ket{b(\theta)}}{ \lVert \tilde{H}^{-1}\ket{b(\theta)} \rVert},
\end{equation}
and the measurement of the observable $\tilde{H}$ on this output state (Step 3.(b)), due to the following relation:
\begin{equation} \label{eq:expt_relation}
\begin{split}
    \bra{b(\theta)} \tilde{H}^{-1} \ket{b(\theta)} &= \lVert \tilde{H}^{-1}\ket{b(\theta)} \rVert^2 \bra{x} \tilde{H}^{\dag} \tilde{H}^{-1} \tilde{H} \ket{x}\\
    &= \lVert \tilde{H}^{-1}\ket{b(\theta)} \rVert^2 \bra{x} \tilde{H} \ket{x},
\end{split}
\end{equation}
where we have used the self-adjointness of $\tilde{H}$. The normalization factor $\lVert \tilde{H}^{-1}\ket{b(\theta)} \rVert^2$ can be extracted separately, depending on the implementation of the QLSP solver, and can generally be expressed as an easily accessible statistical feature of $\ket{x}$ (see Eq.~(\ref{eq:norm_expt}) later). We can further determine an upper limit on the precision required for the computation of $C(\theta)$ by imposing that it must be able to resolve and distinguish between the target eigenstate $\ket{\lambda_k}$ and those of its immediate neighbors $\ket{\lambda_{k \pm 1}}$ -- in other words, to ensure that the variational principle on $\tilde{H}^{-1}$ is respected in the presence of an error in the computation of $C(\theta)$. As we later detail in Section.~\ref{sec:other_components}, this translates to precisions that scale as $O(\Delta)$ for both Step 3.(a) and Step 3.(b), and can be relaxed if a superposition of eigenstates from a fixed energy interval is asked for, instead of targeting an eigenstate. 

Depending on the resources required for the implementation of both steps, the algorithm can either be classified as requiring fault-tolerant error-corrected quantum computers or feasible in near-term devices, which we analyze and describe in the following sections. For the former case, the runtime required by an efficient implementation with the HHL algorithm scales as $\text{poly}(\text{log}(N), s, \kappa, 1/\Delta)$ (where $\text{poly}(...)$ denotes a function that is upper bounded by a polynomial function of its arguments), which we detail in Sections.~\ref{sec:fault_tol} and \ref{sec:fault_tol_anal}. For the latter case, numerical evidence indicates that an implementation based on the Variational Quantum Linear Solver (VQLS) \cite{bravo2019variational} suitable for near-term quantum computers inherits similar scalings, which we detail in Sections.~\ref{sec:near_term} and \ref{sec:near_term_analysis}. The number of qubits required in both cases scale efficiently with system size as $O(\text{log}(N))$. 

More generally, the ability to evaluate the expectation value of the inverse of a Hermitian operator via Eq.~(\ref{eq:expt_relation}) is relevant beyond its usage as the cost function of a variational procedure in our algorithm. For instance, the same quantity appears widely in studies of the dynamics and high temperature properties of strongly correlated quantum systems via Green's function. We expand on this utility, and describe further examples in other areas in physics and chemistry that can benefit from the ability to efficiently prepare highly excited eigenstates of large systems in Appendix.~\ref{appendix:applications}.

We delegate further technical details of the procedure beyond the QLSP subroutine (Step 3.(a)) to Section.~\ref{sec:other_components}, including an analysis of the overall error bounds of Step 3, the expectation estimation and norm computation subroutines (Step 3.(b)), the choice of the constant $r$ (Step 1), and issues regarding the dependence on the spectral gap $\Delta$ at large $n$.

The choice of the variational ansatz $U(\theta)$ that prepares the input state via Eq.~(\ref{eq:input_state}) is also crucial. As mentioned before, this is especially true in our context where the target eigenstate has a clear physical structure that can be exploited, and allows data mapping/loading of the form $\vec{b} \rightarrow \ket{b}$ to prepare the QLSP solver's input to be avoided. Besides having a depth that is at most $\text{poly}(n)$, it should be chosen to be expressible enough such that the target eigenstate is located within the set of quantum states reachable, yet avoid trainability issues arising from barren plateaus \cite{mcclean2018barren,cerezo2021cost}. Motivated initialization strategies and adaptive ansatzes that can potentially circumvent this issue will be relevant \cite{grant2019initialization,grimsley2022adapt}. We provide further elaborations on these points in Section.~\ref{sec:advantage}.

Finally, we point to Sections.~\ref{sec:relation_quantum} and ~\ref{sec:advantage} for an evaluation of the performance and relation of the proposed algorithm, compared to existing quantum and classical algorithms for the same task. We describe key features possessed by our algorithm that circumvent limitations of existing variational quantum algorithms in Section.~\ref{sec:relation_quantum}, and elaborate on situations and conditions under which this algorithm can outperform existing classical implementations in Section.~\ref{sec:advantage}, which is dependent on the whether efficient representation of physical eigenstates and simulation of $H$ on quantum computers can be achieved. This leads us to the conclusion that promising candidates can be found from problems arising from physical contexts such as chemistry, condensed matter, and high-energy physics, and in the regime of large $n$, when the exponential improvement in the scaling in $n$ over classical solvers are able to outscale their polynomially worse scalings in other parameters.

\section{Implementations} \label{sec:implementations}
In the following section, we describe and analyze the resource requirements of two implementations of the algorithm. The \textit{fault-tolerant implementation} involves deep circuits to execute the HHL algorithm, which necessitates fault-tolerant error-corrected devices. The \textit{near-term implementation} is short-depth and noise tolerant, and can therefore be executed on currently available NISQ devices. In our initial discussions on both implementations, we focus on the key (and most demanding) step of the procedure -- the solution of the QLSP for Eq.~(\ref{eq:qlsp}) at Step 3.(a) -- and delegate discussion on the rest (expectation estimation, norm computation, choice of shift parameter, etc.) to Section.~\ref{sec:other_components}.

\subsection{Fault-tolerant implementation} \label{sec:fault_tol}
With access to fault-tolerant error-corrected quantum computers, a plethora of QLSP solvers with rigorous performance guarantees can be employed to obtain the solution of the QLSP Eq.~(\ref{eq:qlsp}). For concreteness, we choose the well-studied HHL algorithm, adopting similar assumptions as HHL \cite{harrow2009quantum}. It involves time evolution $e^{-iHt}$ as the input model for $H$, which is efficient for local Hamiltonians of the form Eq.~(\ref{eq:k-local}) considered in our work (via e.g. Trotter decompositions \cite{lloyd1996universal}). Beyond the HHL algorithm, improved QLSP solvers have been found and should be used in practice, involving techniques such as linear combination of unitaries (LCU) \cite{berry2015hamiltonian,childs2017quantum}, quantum singular value transformation (QSVT) \cite{low2017optimal,gilyen2019quantum} and adiabatic quantum computing \cite{lin2020optimal,an2022quantum,costa2022optimal}. Across these variants, the logarithmic scaling in $N$ and polynomial scaling in $\kappa$ and $s$ remain qualitatively similar, while the scaling in $1/\epsilon_{\text{HHL}}$ can be exponentially improved. We provide a review of the HHL algorithm and establish its notations in Appendix.~\ref{appendix:hhl}, with its modified quantum circuit displayed in Fig.~(\ref{fig:si_hhl}).

To adapt the HHL algorithm for shift-inversion, the QPE step is performed as per usual, while the eigenvalue inversion subroutine -- consisting of a $R_\text{y}(\phi)$ Pauli rotation controlled by qubits of the eigenvalue register with $\phi = \text{arccos}(C/\tilde{\lambda})$ -- is modified by choosing $\phi = \text{arccos}(C/(\tilde{\lambda} - \sigma))$. This step can be implemented efficiently with standard arithmetic circuits, using resources that scale only polynomially with precision (which we hide) \cite{lin2022lecture,vazquez2022enhancing}. This results in the transformation:
\begin{equation}
\begin{split}
    \sum_i \beta_i \altket{\tilde{\lambda}_i} \ket{\lambda_i} \ket{0} \rightarrow & \sum_i \beta_i \sqrt{1 -  \frac{C^2}{(\tilde{\lambda}_i - \sigma)^2}} \altket{\tilde{\lambda}_i} \ket{\lambda_i} \ket{0} \\
    & + \beta_i \frac{C}{\tilde{\lambda}_i-\sigma} \altket{\tilde{\lambda}_i} \ket{\lambda_i} \ket{1},
\end{split}
\end{equation}
which is followed by the uncomputation of the eigenvalue register. Post-selecting the `1' states of the ancilla qubit, which succeeds with probability:
\begin{equation} \label{eq:proba}
    p_1 = \left\Vert \sum_i C \frac{\beta_i}{\tilde{\lambda}_i - \sigma} \ket{\lambda_i} \right\Vert^2 \approx C^2 \| \tilde{H}^{-1} \ket{b(\theta)} \|^2,
\end{equation} 
then yields:
\begin{equation}
    \frac{1}{\| \tilde{H}^{-1} \ket{b(\theta)} \|} \sum_i \frac{\beta_i}{\tilde{\lambda}_i-\sigma} \altket{\tilde{\lambda}_i} \equiv \ket{\tilde{x}},
\end{equation}
which approximates the solution Eq.~(\ref{eq:qlsp}). The choice of the constant $C$ must also be modified accordingly to satisfy unitarity while at the same time remain large to maximize the success probability. This results in:
\begin{equation} \label{eq:gap}
    C = \min_i \{ |\lambda_i - \sigma|\} = \min \{ r\Delta, (1-r)\Delta \}.
\end{equation}
Due to Eq.~(\ref{eq:proba}), the norm $\| \tilde{H}^{-1} \ket{b(\theta)} \|$ can also be directly computed by first estimating the success probability $p_1$ via sampling, in addition to methods described in Section.~\ref{sec:other_components} (such as via Eq.~(\ref{eq:norm_expt})).

This step needs to be executed up to error $\epsilon_{\text{HHL}} = O(\Delta)$ in order for the target eigenstate and its neighbors to be distinguished (following the error analysis to be detailed in Section.~\ref{sec:other_components}, so that Eq.~(\ref{eq:epsilon_condition}) is satisfied). Combining this requirement with the complexity of the HHL algorithm (Section.~\ref{sec:fault_tol_anal}, Eq.~(\ref{eq:ft_depth})), we find a final runtime of $\text{poly}(\text{log}(N), s, \kappa, 1/\Delta)$.

Focusing on the spectral gap $\Delta$, the depth of the HHL circuit scales as $O(1/\Delta)$, while $1/p_1$ scales as $O(1/\Delta^2)$, resulting in an overall $O(1/\Delta^3)$ runtime. Notably, the position of the target eigenstate in the spectrum does not directly affect the resulting complexity, unlike existing approaches based on iterative orthogonalization \cite{santagati2018witnessing,higgott_variational_2019,nakanishi2019subspace} (c.f. Section.~\ref{sec:relation_quantum}). Rather, the complexity of this step is determined by the local spectral density near $\lambda_k$. Amplitude amplification is expected to further improve $p_1$ by a quadratic factor, to $\Omega(\Delta)$ in the worst case. 

Notably, in the context of our algorithm where variational optimization of the shift-inverted cost function Eq.~(\ref{eq:cost_func}) leads to the convergence of $\ket{b(\theta)}$ to the target eigenstate, $p_1$ increases as optimization progresses, tending to unity if the target eigenstate can be represented exactly by the variational ansatz (from Eq.~(\ref{eq:proba}), $p_1$ approaches unity as the input quantum state $\ket{b(\theta)}$ concentrates at the target eigenstate, i.e. $\beta_k \rightarrow 1$). In practice, this alleviates the dependence on $\Delta$ as optimization progresses, so the runtime of this step will typically be much better than its worst case complexity. Conversely, the value of $p_1$ can be used as a `witness' for the closeness of $\ket{b(\theta)}$ to the intended solution, and therefore as an alternative cost function to be maximized.

For the task of estimating statistical quantities by sampling from the solution of a QLSP (as is our case, for the computation of $\bra{b(\theta)} \tilde{H}^{-1} \ket{b(\theta)}$ via Eq.~(\ref{eq:expt_relation})), it is highly unlikely that the linear dependence on $\kappa$ (which translates to a dependence in $1/\Delta$) and the polynomial dependence on $1/\epsilon$ can be substantially improved, due to complexity theoretic reasons \cite{harrow2009quantum,harrow2014review}. This indicates that for this step, the regime where implementations based on QLSP solvers can outperform classical implementations will likely be for large $N$ (and in the absence of classical approximative schemes such as tensor network approximations), when the QLSP solvers' exponentially improved scaling in $N$ outscale their polynomially worse scalings in $\kappa$ and $\epsilon$ compared to classical solvers. We expand on this discussion, along with the implications of rapidly vanishing $\Delta$ in Section.~\ref{sec:advantage}.

\subsection{Analysis of the fault-tolerant implementation} \label{sec:fault_tol_anal}

The following result bounds the resources required to output an approximation to $\ket{x}$, such that $\bra{b(\theta)} \tilde{H}^{-1} \ket{b(\theta)}$ can be computed up to an error $\epsilon_{\text{HHL}}$ via Eq.~(\ref{eq:expt_relation}) (Step 3). 

The proof follows from a direct extension of \cite{harrow2009quantum} to account for the presence of the shift $\sigma$ which introduces a dependence on $\Delta$. Similar assumptions on the usage of filter functions for preconditioning and the implementations of Hamiltonian simulation and QPE are adopted. Details on obtaining the complexities of intermediate steps also remain similar, and can be improved by considering recent developments in QLSP solvers \cite{berry2015hamiltonian,childs2017quantum,low2017optimal,gilyen2019quantum,lin2020optimal,an2022quantum,costa2022optimal} and Hamiltonian simulation \cite{berry2015hamiltonian,low2017optimal}. We omit a lengthy discussion, focusing instead on the effect of introducing a shift $\sigma$ to the procedure.

\begin{theorem*}
\label{theorem:error}
Consider the quantum circuit defined in Section.~(\ref{sec:fault_tol}) to execute Step 3.(a) via the HHL algorithm. To output a quantum state $\ket{\tilde{x}}$ such that the error in its expectation value for the observable $\tilde{H}$ due to the imprecision of the HHL algorithm is within a fixed error $\epsilon_{\text{HHL}}$, i.e.: 
\begin{equation}
    | \bra{\tilde{x}}\tilde{H}\ket{\tilde{x}} - \bra{b(\theta)}\tilde{H}^{-1}\ket{b(\theta)} | \leq \epsilon_{\text{HHL}},
\end{equation}
the depth of the circuit scales as:
\begin{equation} \label{eq:ft_depth}
    O(\text{log}(N) s^2 \kappa/\epsilon_{HHL} + \text{polylog}(N)),
\end{equation}
with a success probability $p_1$ that scales as $\Omega(\Delta^2)$ in the worst case and up to $O(1)$ in the best case, and requires $\text{log}(N) + 1 + O(\text{log}(1/\epsilon_{\text{HHL}}))$ qubits.
\end{theorem*}

\begin{proof} 
Due to the HHL algorithm, the main source of error is contributed by the QPE step. It involves querying $H$ via Hamiltonian simulation $e^{-iHt}$ for time $t \leq t_0$, which can be implemented in depth $T = O(\text{log}(N) s^2 t_0)$ for an $s$-sparse $H$ \cite{berry2007efficient}, where the dependence in precision is omitted as it scales much more weakly than other parameters.

Firstly, adapting from \cite{harrow2009quantum}, taking $t_0 = O(\lVert M \rVert_{\infty} \kappa / \epsilon_{\text{HHL}})$ leads to an error of $\epsilon_{\text{HHL}}$ in the expectation value of an observable $M$ due to the imprecision of the HHL algorithm, i.e.:
\begin{equation} \label{eq:m_error_bound}
    |\expt{\tilde{x}}{M}{\tilde{x}} - \expt{x}{M}{x} | \leq  O( \lVert M \rVert_{\infty} \frac{\kappa}{t_0}  ).
\end{equation}
To show this, note that since the absolute difference between the expectation value of an observable $M$ for quantum states $\ket{\psi}$ and $\ket{\phi}$ is bounded above by:
\begin{equation} \label{eq:m_fidelity_bound}
    |\expt{\psi}{M}{\psi} - \expt{\phi}{M}{\phi} | \leq 2 \lVert M \rVert_{\infty} \sqrt{1-|\braket{\psi}{\phi}|^2},
\end{equation}
it suffices to bound the overlap between $\ket{\tilde{x}}$ and $\ket{x}$. From \cite{harrow2009quantum}, we have:
\begin{equation}
    Re \braket{\tilde{x}}{x} \geq 1-O(\kappa^2/t_0^2),
\end{equation}
which leads to the final bound Eq.~(\ref{eq:m_error_bound}) via Eq.~(\ref{eq:m_fidelity_bound}).

The observable $M$ to be measured from the output quantum state $\ket{\tilde{x}}$ is always taken to be equal to $\tilde{H}$, resulting in Eq.~(\ref{eq:expt_relation}) in the absence of any errors. Since $\sigma$ is always chosen between 0 and the largest eigenvalue of $H$, the shifted operator have eigenvalues in $(-1, 1)$, so that $\lVert M \rVert_{\infty} = 1$. Eq.~(\ref{eq:m_error_bound}) then implies that an error of $\leq \epsilon_{HHL}$ will require a simulation time of up to $t_0 = \kappa/\epsilon_{HHL}$. The depth of QPE therefore scales as $O(\text{log}(N) s^2 \kappa/\epsilon_{HHL})$. Since the preparation of $\ket{b(\theta)}$ is assumed to scale as $O(\text{polylog} N)$ with the choice of an efficient variational ansatz $U(\theta)$, the final circuit depth is $O(\text{log}(N) s^2 \kappa/\epsilon_{HHL} + \text{polylog}(N))$. 

The success probability of post-selecting the shift-inverted quantum state, $p_1$, depends on both $\Delta$ and $\| \tilde{H}^{-1} \ket{b} \|$ through Eqs.~({\ref{eq:proba}}) and (\ref{eq:gap}). In the worst case where $\| \tilde{H}^{-1} \ket{b} \| \geq 1$ is saturated, $p_1 \sim \Delta^2$. In the best case, $\| \tilde{H}^{-1} \ket{b} \| \leq 1/\Delta$ is saturated and $p_1 \sim 1$, which occurs when the input quantum state $\ket{b}$ concentrates at the target eigenstate.

For the number of qubits, the ancillary register always requires 1 qubit, and the solution register contributes $n = \text{log}(N)$ qubits. Finally, the size of the eigenvalue register $m$ determines the binary precision to which the eigenvalues of the Hamiltonian are represented, and uses $m = O(\text{log}(1/\epsilon_{\text{HHL}}))$ qubits.

\end{proof}

\subsection{Near-term implementation} \label{sec:near_term}
Currently available near-term quantum computers are characterized by limited qubit coherence times and relatively high error rates, which limits the circuit depth of executable quantum circuits. As an alternative to QLSP solvers such as the HHL algorithm that require fault-tolerant error-correction, recently developed near-term QLSP solvers \cite{subacsi2019quantum,bravo2019variational,xu2021variational,huang2021near} possess milder resource requirements that are consistent with these limitations. For concreteness, we describe an implementation with the Variational Quantum Linear Solver (VQLS) \cite{bravo2019variational}. Together with a suitable expectation estimation algorithm such as operator averaging (described in Section.~\ref{sec:other_components}), this enables the computation of inverse expectation values (Step. 3) on near-term quantum computers.

VQLS effectively casts the QLSP as a ground state optimization problem and prepares the solution $\ket{x}$ in a variational manner. Given an input matrix $A$ decomposable as a linear combination of $L$ unitaries $A_i$ with complex coefficients $c_i$:
\begin{equation} \label{eq:vqls_lincomb}
    A = \sum_{i=1}^L c_i A_i
\end{equation}
and a quantum state $\ket{b} = U \ket{0...0}$ prepared by a unitary circuit $U$, define the Hamiltonian $G$, which can be chosen to be either:
\begin{equation} \label{eq:g_global}
    G_{\text{global}} = A^\dag (\mathbb{1} - U\ket{0...0} \bra{0...0}U^{\dag}) A
\end{equation}
or
\begin{equation} \label{eq:g_local}
    G_{\text{local}} = A^\dag U \left( \mathbb{1} - \frac{1}{n} \sum_{i=1}^n \ket{0_i}\bra{0_i} \otimes \mathbb{1}_{\overline{i}}  \right) U^\dag A.
\end{equation}
It can be shown that the unique ground state of $G$ with energy 0 corresponds to the quantum state $\ket{x}$ proportional to $A^{-1} \ket{b}$ \cite{subacsi2019quantum}, where $\mathbb{1}_{\overline{i}}$ denotes the identity operator acting on all qubits besides the $i$-th. Denoting $\ket{x(\phi)} = V(\phi) \ket{0...0}$ the quantum state prepared by the variational ansatz $V(\phi)$ and $\phi$ its parameters, VQLS then consists of solving the optimization problem:
\begin{equation} \label{eq:opt_prob_vqls}
    \phi^* = \argmin_{\phi} \bra{x(\phi)} G \ket{x(\phi)}
\end{equation}
in a variational manner to yield $\ket{x} \approx V(\phi^*)\ket{0...0}$, with the cost function $C_{\text{VQLS}} (\phi) = \bra{x(\phi)} G \ket{x(\phi)} \geq 0$.

$C_{\text{VQLS}}(\phi)$ can be computed by evaluating each of its $O(L^2)$ terms with variants of the Hadamard test, in a manner feasible for near-term quantum computers. A termination condition for the value of $C_{\text{VQLS}} (\phi)$ to achieve a fixed error tolerance can also be defined explicitly in terms of the parameters of the input matrix. Finally, existing numerical scaling results for the VQLS reveal runtimes that are comparable to QLSP solvers such as the HHL algorithm, including polynomial scalings in $n$, $\kappa$, $s$ and $\epsilon$. 

Within our algorithm, the input matrix corresponds to the shifted Hamiltonian $\tilde{H}$, which is a $k$-local Hamiltonian of the form Eq.~(\ref{eq:k-local}). It therefore takes the form of Eq.~(\ref{eq:vqls_lincomb}) naturally, since the Pauli terms $h_i$ are unitary. We can further reduce the number of terms in the cost function by a factor of $\sim n$ by grouping anti-commuting Pauli terms together \cite{izmaylov2019unitary,zhao2020measurement,ralli2021implementation}, which amounts to solving a minimum clique cover problem for a graph representing the anti-commutation of the Pauli terms. On the other hand, $\ket{b}$ corresponds to the parametrized quantum state $\ket{b(\theta)} = U(\theta) \ket{0...0}$ at each step of the optimization over $\theta$. Finally, following the argument from Section.~\ref{sec:other_components}, we can further define a termination condition for $C_{\text{VQLS}} (\phi)$ in terms of the parameters $\kappa$ and $\Delta$ of $H$, by imposing that the computation of $\bra{b(\theta)} \tilde{H}^{-1} \ket{b(\theta)}$ via the output of VQLS through Eq.~(\ref{eq:expt_relation}) is sufficiently accurate to distinguish between the shift-inverted energies of the target's neighboring eigenstates (Eq.~(\ref{eq:epsilon_condition})).

With this implementation, the overall algorithm now consists of two optimization loops, with an outer loop over $\theta$ that searches for the target excited eigenstate (by solving Eq.~(\ref{eq:opt_prob})) and an inner loop over $\phi$ that allows the computation of inverse expectation values (Step.3, by solving Eq.~(\ref{eq:opt_prob_vqls})). We refer the reader to Section.~(\ref{sec:near_term_analysis}) for a detailed elaboration on the runtime of the procedure, including discussions on trainability and barren plateaus.

Notably, this approach bears similarities to an existing classical shift-invert-DMRG-based method used to obtain state-of-the-art results for MBL Hamiltonians \cite{yu2017finding}. There, in order to avoid costly intermediate steps involving the solution of linear systems, a variational optimization is performed over Matrix Product States (MPS) $\ket{x_{\text{MPS}}}$ to minimize the distance $\lVert \tilde{H}\ket{b} - \ket{x_{\text{MPS}}} \rVert^2$. In this context, the VQLS can be viewed as the quantum analogue to this classical optimization procedure, with the MPS $\ket{x_{\text{MPS}}}$ replaced by variational quantum states $\ket{x(\phi)}$, and the evaluation of the distance replaced by cost functions involving the Hamiltonian $G$. For the case of MBL Hamiltonians, the effectiveness of the classical tensor network approach of \cite{yu2017finding} indicates that the QLSP solution Eq.~(\ref{eq:qlsp}) can be efficiently represented as MPS. This further implies that our implementation based on low-depth, low-entanglement choices of the ansatz $V(\phi)$ can also efficiently represent $\ket{x}$ on quantum computers in this case. We also elaborate on this observation in Section.~(\ref{sec:applications_spin}).

\subsection{Analysis of the near-term implementation} \label{sec:near_term_analysis}
We elaborate on the runtime of the near-term implementation with VQLS, which depends on the whether both the computation and optimization of the cost function $C_{\text{VQLS}}(\phi)$ can be performed efficiently. Full descriptions of the circuits used can be found from \cite{bravo2019variational}. We also note that the number of qubits required is at most $2n+1$ for all circuits considered.

Firstly, we define a termination condition for the inner loop optimization over $\phi$ in terms of the parameters of the problem. Denote $\phi^\star$ the output parameters of VQLS (which may not correspond to the global minima $\phi^*$ of Eq.~(\ref{eq:opt_prob_vqls}) due to suboptimal optimization), and $\epsilon_{\text{M}}(\phi) = |\expt{x(\phi)}{M}{x(\phi)} - \expt{x}{M}{x} |$ the error in the expectation value of the observable $M$ for VQLS' output quantum state $\ket{x(\phi)}$. Following a similar argument as Section.~(\ref{sec:other_components}), we require an output quantum state such that its expectation value for the observable $M = \tilde{H}$ is accurate enough to distinguish between the energy of the target eigenstate and those of its immediate neighbors via Eq.~(\ref{eq:epsilon_condition}), resulting in the condition $\epsilon_{\text{M}}(\phi^\star) < \frac{r}{1-r}\Delta$. Since the following bounds hold \cite{bravo2019variational}:
\begin{alignat}{2}
    C_{\text{VQLS}}(\phi) &\geq \frac{\epsilon_M^2 (\phi)}{4 \lVert M \rVert_{\infty}^ 2 \kappa^2} &&\quad \text{if} \quad  G = G_{\text{global}}, \\
    C_{\text{VQLS}}(\phi) &\geq \frac{\epsilon_M^2 (\phi)}{4\lVert M \rVert_{\infty}^2 \kappa^2} \frac{1}{n} &&\quad \text{if} \quad  G = G_{\text{local}},
\end{alignat}
we can ensure the satisfaction of the aforementioned condition by terminating the optimization (over $\phi$) only when $C_{\text{VQLS}}(\phi) \leq C_{\text{min}}$ is satisfied, with the choices:
\begin{alignat}{2}
    C_{\text{min}} &= \frac{r^2 \Delta^2}{4 (1-r)^2 (\kappa-1)^2} &&\quad \text{if} \quad  G = G_{\text{global}}, \\
    C_{\text{min}} &= \frac{r^2 \Delta^2}{4 (1-r)^2 (\kappa-1)^2} \frac{1}{n} &&\quad \text{if} \quad  G = G_{\text{local}},
\end{alignat}
where we have used $\lVert M \rVert_{\infty} = 1-1/\kappa$.

Next, we consider the runtime to compute $C_{\text{VQLS}} (\phi)$, which depends on the number of terms it contains and the depth of the circuits used to evaluate them individually. For an input matrix that is a sum of $L$ unitaries, the number of terms in the cost function scales as $O(L^2)$ for the choice $G = G_{\text{global}}$ (Eq.~(\ref{eq:g_global})) and $O(nL^2)$ for $G = G_{\text{local}}$ (Eq.~(\ref{eq:g_local})). In our context where the input matrix corresponds to physical and chemical Hamiltonians $\tilde{H}$ of the form Eq.~(\ref{eq:k-local}), $L$ generically scales polynomially in $n$, e.g. as $O(n^4)$ for molecular Hamiltonians and as $O(nD)$ for spin models on a $D$-dimensional lattice. This scaling can be directly reduced by a factor of $\sim n$ by employing term grouping strategies \cite{izmaylov2019unitary,zhao2020measurement,ralli2021implementation} that group anti-commuting Pauli terms together, which amounts to solving a minimum clique cover problem for a graph representing the anti-commutation of the Pauli terms. While this task is NP-hard in general, polynomial-time heuristics such as the Recursive Largest First heuristic \cite{leighton1979graph} has been shown to be effective \cite{izmaylov2019unitary}. On the other hand, each term in $C_{\text{VQLS}}(\phi)$ can be computed in a parallel manner with the Hadamard test or its variants. Denoting $D_U$ the depth of the circuit $U(\theta)$ (which prepares VQLS' input state $\ket{b(\theta)}$) and $D_V$ the depth of the ansatz $V(\phi)$ (which parametrizes the solution $\ket{x(\phi)}$), the depths of the Hadamard test circuits are dominated by $D_U$ and $D_V$, and scales as $O(D_U) + O(D_V)$. Since we take $D_U$ to scale polynomially with $n$ with physically motivated choices of $U(\theta)$ (see also Section.~\ref{sec:advantage}), the depth of the cost computation circuits rests on the choice of the ansatz $V(\phi)$.

This brings us finally to the issue of trainability, which, together with the cost of computing $C_{\text{VQLS}}(\phi)$, determines the total runtime of VQLS. We remark that precisely bounding the runtime of the optimization process for heuristics such as the VQLS is generally intractable analytically, and depends on the choices of the ansatz, cost function, classical optimization algorithm, and parameter initialization scheme. This is further complicated by the possible appearance of barren plateaus that limit the trainability of variational algorithms due to the form the ansatz \cite{mcclean2018barren,huang2021near}, the form of the cost function \cite{cerezo2021cost}, and the presence of noise \cite{wang2021noise}. We only attempt to provide indications from numerical scaling results on the total runtime of VQLS obtained by \cite{bravo2019variational} and \cite{xu2021variational}. In particular, \cite{bravo2019variational} were able to extract polynomial scalings in parameters $n$, $\kappa$, $s$ and $\epsilon$ comparable to the HHL algorithm. Among other examples, these scalings were obtained for Ising Hamiltonians with nearest-neighbor interactions of increasing condition number, with $V(\phi)$ chosen to be the hardware-efficient ansatz (consisting of nearest-neighbor interactions at each layer), and $G_{\text{local}}$ chosen for $C_{\text{VQLS}}(\phi)$ (which exhibits improved trainability due to its local structure \cite{bravo2019variational,cerezo2021cost}). Combined with successful implementations of small-scale instances on real quantum computers, these results provide positive indications on the runtime of the near-term implementation in our context.

Notably, the above discussions indicate that practical implementations for spin systems defined on low-dimensional lattices will be attractive and feasible in the near term, due to the mild $O(nD)$ scaling in the number of Hamiltonian terms, which can further be improved via term grouping strategies. Section.~\ref{sec:applications_spin} describes numerical simulation results for such a problem and provides further elaborations.

Finally, we point out that physically motivated choices of $V(\phi)$ will likely be crucial to alleviate the barren plateau problem, an example being the Alternating Operator Ansatz \cite{subacsi2019quantum,huang2021near} and the ADAPT-VQE ansatz that have been shown to circumvent the phenomenon in certain situations \cite{grimsley2019adaptive,tang2021qubit}.

\subsection{Other components of the procedure} \label{sec:other_components}
This section discusses other components of the proposed algorithm, namely (in the following order) an overall error analysis for the cost computation (Step 3), the expectation estimation subroutine (Step 3.(b)), the computation of the constant of normalization $\lVert \tilde{H}^{-1}\ket{b(\theta)} \rVert$, issues regarding the spectral gap $\Delta$, and trade-offs in the choice of the hyperparameter $r$ (Step 1).

We begin by discussing the overall sources of errors and required precision of the algoritm's subroutines. Denote by $\epsilon$ the error accrued in the computation of the cost function Eq.~(\ref{eq:cost_func}) at Step 3 via the relation Eq.~(\ref{eq:expt_relation}):
\begin{equation} \label{eq:error_main}
    \epsilon \equiv  | \exptshort{b(\theta)}{\tilde{H}^{-1}}{b(\theta)} - \lVert \tilde{H}^{-1}\ket{b(\theta)} \rVert^2_{\text{est}} \bra{\tilde{x}} \tilde{H} \ket{\tilde{x}}_{\text{est}} |
\end{equation}
where the subscript $\textit{est}$ indicates an estimate of the quantity, and $\ket{\tilde{x}}$ is the imprecise output of the QLSP solver. Assuming that the norm $\lVert \tilde{H}^{-1}\ket{b(\theta)} \rVert$ is computed without errors for now, the main sources of errors contributing to $\epsilon$ are $\epsilon_{\text{HHL}}$, the error from the imprecise output of the HHL algorithm:
\begin{equation}
    \epsilon_{\text{HHL}} \equiv |\exptshort{\tilde{x}}{\tilde{H}}{\tilde{x}} - \exptshort{x}{\tilde{H}}{x} |,
\end{equation}
and $\epsilon_{\text{exp}}$, the error from the imprecise estimation of the expectation value of $\tilde{H}$ on the state $\ket{\tilde{x}}$:
\begin{equation}
    \epsilon_{\text{exp}} \equiv |\exptshort{\tilde{x}}{\tilde{H}}{\tilde{x}}_{\text{est}} - \exptshort{\tilde{x}}{\tilde{H}}{\tilde{x}}|.
\end{equation}
Developing Eq.~(\ref{eq:error_main}) yields:
\begin{equation}
\begin{split}
    \epsilon &= \lVert \tilde{H}^{-1}\ket{b(\theta)} \rVert^2 | \exptshort{x}{\tilde{H}}{x} - \bra{\tilde{x}} \tilde{H} \ket{\tilde{x}}_{\text{est}} | \\
    & \leq \lVert \tilde{H}^{-1}\ket{b(\theta)} \rVert^2 \left( \epsilon_{\text{exp}} + \epsilon_{\text{HHL}} \right),
\end{split}
\end{equation}
where the final inequality follows from the triangle inequality. As mentioned in Section.~\ref{sec:algorithm}, we impose the condition that $\epsilon$ must be sufficiently small for $C(\theta)$ to be able to distinguish between the target eigenstate $\ket{\lambda_k}$ and those of its immediate neighbors $\ket{\lambda_{k \pm 1}}$, in order for the minimization of the estimated value of $C(\theta)$ to yield $\ket{\lambda_k}$ uniquely. This results in the condition:
\begin{equation} \label{eq:epsilon_condition}
\begin{split}
    \epsilon 
    & \leq \left| \frac{1}{\lambda_k - \sigma} - \frac{1}{\lambda_{k \pm 1} - \sigma} \right| \\
    & = \frac{1}{\Delta r (1-r)}.
\end{split}
\end{equation}
Noting that the squared norm is bounded above as:
\begin{equation}
\begin{split}
    \lVert \tilde{H}^{-1}\ket{b(\theta)} \rVert^2 &= 
    \sum_i \left( \frac{\beta_i}{\lambda_i - \sigma} \right)^2 \\
    & \leq 1/(r \Delta)^2,
    \end{split}
\end{equation}
this translates to the conditions that both $\epsilon_{\text{exp}}$ and $\epsilon_{\text{HHL}}$ must be of the order $O(\Delta)$ for Eq.~(\ref{eq:epsilon_condition}) to hold.

Once the solution $\ket{\tilde{x}}$ of the QLSP is prepared via either the fault-tolerant (Section.~\ref{sec:fault_tol}) or near-term (Section.~\ref{sec:near_term}) implementations, it remains to measure the observable $\tilde{H}$ to complete Step 3, through Eq.~(\ref{eq:expt_relation}). The runtime of this step scales polynomially with the inverse error $1/\epsilon_{\text{exp}}$, and can be achieved with several standard approaches. Firstly, operator averaging \cite{peruzzo2014variational} contributes at most $O(L / \epsilon_{\text{exp}}^2 )$ circuit repetitions, where $L$ (the number of terms in the Hamiltonian Eq.~(\ref{eq:k-local})) generically scales as $O(\text{polylog}(N))$ for applications in chemistry and many-body physics, and can be greatly reduced by parallelization and term grouping strategies \cite{arrasmith2020operator,huggins2021efficient,tilly2022variational}. This approach does not require additional ancillary qubits and circuit depth, and is therefore compatible with both implementations. On the other hand, QPE-based methods are able to approach the optimal Heisenberg scaling of $O(1/\epsilon_{\text{exp}})$, at the expense of requiring additional ancillary qubits and $O(1/\epsilon_{\text{exp}})$ calls to Hamiltonian simulation \cite{knill2007optimal}. This is thus the preferred method when fault-tolerant, error-corrected quantum computers are available. Also noteworthy are recently developed schemes with resource scalings that are intermediate between the two strategies \cite{wang2019accelerated}.

Next, for the constant of normalization $\lVert \tilde{H}^{-1}\ket{b(\theta)} \rVert$, applying the triangle inequality to Eq.~(\ref{eq:error_main}) together with $\lVert \tilde{H}^{-1} \rVert_{\infty} \leq 1$, and imposing Eq.~(\ref{eq:epsilon_condition}) in a similar manner as $\epsilon_{\text{HHL}}$ and $\epsilon_{\text{exp}}$ leads to the condition:
\begin{alignat}{2}
   \epsilon_{\text{norm}} &\equiv | \lVert \tilde{H}^{-1}\ket{b(\theta)} \rVert^2 | - \lVert \tilde{H}^{-1}\ket{b(\theta)} \rVert^2 |_{\text{est}} | \\
    &= O(\Delta).
\end{alignat}
The norm can be computed in several ways. Most straightforwardly, developing $\lVert \tilde{H}^{-1}\ket{b(\theta)} \rVert$ yields:
\begin{equation} \label{eq:norm_expt}
	\lVert \tilde{H}^{-1}\ket{b(\theta)} \rVert^2 = \frac{1}{\exptshort{x}{\tilde{H}^2}{x}},
\end{equation}
so that it can be evaluated with expectation estimation algorithms of the preceding paragraph. In the case of the HHL implementation described in Section.~(\ref{sec:fault_tol}), it can alternatively be computed via its relation with the success probability Eq.~(\ref{eq:proba}).

The spectral gap $\Delta$ in the bulk of generic many-body systems vanishes exponentially with system size $n$, which results in a potentially exponential scaling in problem size due to the $\text{poly}(1/\Delta)$ runtime dependence. Indeed, the same issue arises in techniques (both classical and quantum) based on spectral transformations \cite{pietracaprina2018shift,sierant2020polynomially}. We discuss several methods to alleviate this dependence. Firstly, it is known that information on the structure of eigenstates and the presence of symmetries can be taken into account to resolve eigenstates that are exponentially close in energy. This idea is exploited in classical simulations of large MBL Hamiltonians (discussed in Section.~\ref{sec:applications_spin}). By leveraging on the fact that neighboring eigenstates (in energy) differ extensively in their spatial structure due to a set of conserved local integrals of motion, they can be prevented from mixing despite exponentially vanishing many-body level spacings \cite{khemani2016obtaining,devakul2017obtaining}. More generally, we refer the reader to \cite{baiardi2019optimization}, where similar ideas have been explored in the context of chemistry (under \textit{root homing methods}), when e.g. knowledge on the structure of electronic and vibrational eigenstates is available. Such information will likely be encoded in the ansatz $U(\theta)$ in our context. Secondly, the $O(\Delta)$ condition on $\epsilon_{\text{HHL}}$ can be relaxed (to constant) if instead of targeting an eigenstate, we accept a superposition of eigenstates from a fixed window of energy. This situation arises when only coarse-grained, global information on the spectrum of a Hamiltonian is required, an example being the study of mobility edges in MBL systems.

Finally, the choice of $r \in (0,1]$ controls the spectral separation between the shift $\sigma$ and the target energy $|\sigma - \lambda_k| = r \Delta$, and is a hyperparameter of the procedure that leads to a trade-off between the cost of solving the QLSP and the difficulty of parameter optimization. Due to the form of the reciprocal function, choosing small values of $r$ close to 0 increases repulsion between $\lambda_k$ and its neighbors, and so penalizes non-target eigenstates more heavily, leading to steeper optimization landscapes. At the same time, the spectral gap of the shifted Hamiltonian at the target energy $|\sigma - \lambda_k| = r \Delta$ is reduced, resulting in increased complexity when solving for Eq.~(\ref{eq:qlsp}). On the other hand, choosing $r$ close to 1 leads to the opposite conclusion of reduced computational complexity in solving for Eq.~(\ref{eq:qlsp}), but more difficult optimization due to milder optimization landscapes. In practice, $r$ can be varied at different stages of the optimization process in a heuristic manner, with smaller values preferred when the optimizer encounters difficulty navigating a flat landscape.

\subsection{Advantages over classical techniques} \label{sec:advantage}
This section discusses relations and possible advantages of our proposed approach over existing classical techniques for the task of preparing highly excited eigenstates deep in the spectrum of a Hamiltonian. We evaluate conditions and situations under which our approach is favourable, where it is expected to outperform classical algorithms based on shift-inversion. 

Before beginning, we remark that the variational nature of our algorithm renders guarantees on the runtime arising from the optimization procedure -- the time it takes to obtain the global optima of the optimization problem of Eq.~(\ref{eq:opt_prob}) -- difficult to obtain, and is not precisely answered in the present work. Such an analysis is generally challenging to carry out in generality, and depends on the choice of the variational ansatz, structure of Hamiltonian, classical optimization algorithm, and form of cost function, and is further complicated by the possible appearance of barren plateaus \cite{mcclean2018barren,wang2021noise}.  Nonetheless, this is not inconsistent with the core findings of our work -- that the shift-inversion procedure can be performed efficiently on quantum computers via QLSPs, and possibly more efficiently than classical methods, under circumstances which we elaborate below. The optimization procedure is also common to classical methods based on shift-inversion, and their success provide positive indications \cite{yu2017finding,baiardi2019optimization}. We find further agreement from limited numerical results at small scales discussed in Appendix.~\ref{appendix:scaling}.

The core idea of adapting classical ground-state eigensolvers to excited states by shift-inversion is well-known in different areas. Depending on the context, it can be exploited in different ways such as with iterative projection methods based on Lanczos/Arnoldi \cite{luitz2015many,pietracaprina2018shift} and Davidson \cite{sleijpen1996jacobi}, and in the presence of approximative ansatzes such as Matrix Product States (MPS) \cite{dorando2007targeted,serbyn2016power,yu2017finding,baiardi2019optimization}. These methods directly target an eigenstate as long as a good estimate of its energy is known, and in some cases are the only viable methods to solve our task for large problem sizes. Indeed, our proposed algorithm can be viewed as an implementation of this idea on quantum computers, with classical linear system solvers and wavefunction representation replaced by its counterparts on quantum computers.


Immediately, the $\text{poly}(1/\Delta)$ runtime (arising from the QLSP solvers' $\text{poly}(\kappa)$ dependence; see e.g. Section.~\ref{sec:fault_tol}, and the proof of Eq.~(\ref{eq:ft_depth})) potentially places a significant overhead for problems with exponentially vanishing spectral gaps $\Delta$ in the bulk. Due to the optimality of the $O(\kappa)$ scaling \cite{harrow2009quantum,harrow2014review}, this rules out implementations of our algorithm on quantum computers with $\text{poly}(n)$ runtime, unless $1/\Delta$ scales only as $\text{polylog}(n)$, which is uncommon as many-body level spacings generically vanish exponentially fast with $n$. Compared to known highly efficient classical algorithms such as conjugate gradient with a $O(\sqrt{\kappa})$ scaling, the best quantum algorithms that achieves the optimal $O(\kappa)$ scaling are at a disadvantage by a polynomial factor. Nonetheless, since even polynomial quantum speed-ups are expected to be practically meaningful \cite{babbush2021focus,lee2023evaluating}, we expect the improved dependence in $n$ of our algorithm (in particular, from Hamiltonian simulation and state preparation, to be discussed later below) to eventually yield an advantage over classical implementations at large $n$.

Beyond that, the main bottlenecks of classical methods based on shift-inversion are twofold. Firstly, the representation of quantum states in classical memory can be prohibitively inefficient when approximative schemes are unavailable, which affects the cost of their storage and manipulations (such as matrix-vector multiplications) in all subsequent computational steps within the algorithm. Secondly, shift-inversion requires obtaining the solution of the linear system Eq.~(\ref{eq:qlsp}), which generically scales exponentially in $n$, and indeed forms a major bottleneck for classical techniques \cite{yu2017finding,pietracaprina2018shift}. Indications to whether our quantum algorithm can outperform them therefore hinges on whether at least one of the above bottlenecks is quantumly easy but classically hard, without imposing additional prohibitive costs. Arguably, physical and chemical Hamiltonians present the best opportunities to satisfy these conditions. 

Firstly, the representation of target eigenstates must be more efficient on quantum computers. Within our algorithm, this condition manifests as the availability of efficient ansatz choices $U(\theta)$ to describe the target eigenstate. Physically motivated ansatzes such as those based on the UCC (and variants more suitable for excited eigenstates such as $k$-UpCCGSD) \cite{lee2018generalized,greene2021generalized} are widely believed to hold strong potential as they can be implemented more efficiently on quantum computers, through e.g. a Trotter decomposition \cite{cao2019quantum,anand2022quantum}. Furthermore, whereas the ground states of local, gapped Hamiltonians are known to admit efficient classical descriptions (such as via MPS) by virtue of their area-law entanglement, the same description generically breaks down for highly excited eigenstates, which may instead exhibit a volume-law \cite{bianchi2022volume}. This prompts explorations of their efficient representation on quantum computers, which shows positive indications in certain cases \cite{evenbly2014class,van2021preparing,hastings2022optimizing}.

Secondly, obtaining the solution of the linear system Eq.~(\ref{eq:qlsp}) must be more efficient on quantum computers. This subroutine appears during the computation of the inverse expectation value (Step 3) within our algorithm, which was found to scale as $\text{poly}(\text{log}(N), 1/\epsilon, s, \kappa, 1/\Delta)$ (Section.~\ref{sec:fault_tol_anal}). Notably, our procedure only involves the extraction of statistical quantities (in our case, the expectation value and possibly the variance) from the solution of the QLSP, and avoids exponentially costly pre/post-processing steps associated with generic linear algebraic applications. The complexity of inverse expectation estimation is therefore dominated by that of the QLSP solver. 

Here, QLSP solvers such as the HHL algorithm are often quoted to possess a significant advantage due to their exponentially improved $\text{poly}(n)$ dependence, despite (polynomially) worse scalings in $\kappa$, $s$, and $\epsilon$. However, this advantage (originating from QPE) is predicated on the fact that Hamiltonian simulation of $H$ (or more generally, the ability of efficiently access $H$, for post-HHL/LCU QLSPs) on quantum computers take $\text{poly}(n)$ time, while classical methods are significantly more expensive. Again, physical and chemical Hamiltonians show promise, as Hamiltonian simulation for a large class of such problems are efficient (via e.g. a Trotter decomposition \cite{lloyd1996universal}). Nonetheless, the possibility that efficient classical solutions for Eq.~(\ref{eq:qlsp}) also exist must be evaluated carefully, which is generally challenging to verify or rule out. In any case, it will likely take large $n$ for the QLSP solvers' improved scaling in $n$ to outscale worse scalings in other parameters, especially for problems with rapidly vanishing $\Delta$.

\subsection{Relation to existing quantum algorithms} \label{sec:relation_quantum}
This section highlights advantages of our proposed approach over existing quantum algorithms for the task of preparing highly excited eigenstates deep in the spectrum of a Hamiltonian, and how it can be used in conjunction with recently developed quantum algorithms based eigenstate filtering.

Our algorithm possesses key features that circumvent limitations of existing variational quantum algorithms. As a spectral transformation, shift-inversion conventionally outperforms \cite{pietracaprina2018shift} simpler methods such as spectral folding/shift-squaring \cite{wang1994solving,peruzzo2014variational,ye2017sigma,santagati2018witnessing}, which consists of optimizing the cost function $\expt{b(\theta)}{(H-\sigma)^2}{b(\theta)}$. While requiring the same inputs with lower computational costs, it is expected to be limited to small-sized problems \cite{pietracaprina2018shift}, due to the second-order polynomial further reducing energy gaps near the target eigenstate.  In contrast, shift-inversion either repels or transforms adjacent eigenstates to the other end of the spectrum, which renders the cost function highly sensitive to mixtures of neighboring eigenstates to facilitate optimization. These features are exploited in classical methods to produce state-of-the-art results beyond the reach of conventional exact diagonalization \cite{serbyn2016power,yu2017finding,pietracaprina2018shift,baiardi2019optimization}. Numerical results from Section.~\ref{sec:applications} and Appendix.~\ref{appendix:additional_numerics} verify this behaviour, where shift-inversion is observed to consistently outperform spectral folding, especially at large system sizes.

On the other hand, variational methods based on iteratively orthogonalizing against lower energy eigenstates \cite{santagati2018witnessing,higgott_variational_2019,nakanishi2019subspace} are ill-adapted for highly excited eigenstates deep in the spectrum, which requires exponentially many lower-energy eigenstates to be found and included in the cost function computation. Instead, our method prepares target eigenstates directly in a targeted manner, without the need to first solve for any lower-energy eigenstates. Our algorithm is therefore expected to be favourable over these approaches for the above task in the regime of moderate to large $n$.

Next, we discuss how our procedure can be used effectively in conjunction with quantum algorithms with precisely known complexities. Recently developed methods based on eigenstate filtering allows the efficient preparation of a target eigenstate, given information on its energy and spectral gap $\Delta$ \cite{ge2019faster,lin2020optimal}. Crucially, they rely on the availability of an input quantum state with high overlap with the target eigenstate. In particular, given an input quantum state $\ket{x_{\text{input}}} = \gamma \ket{\lambda_k} + \ket{\perp}$ where $\braket{\lambda_k}{\perp} = 0$, the complexity of filtering-based algorithms (in terms of queries to $H$ and preparations of $\ket{x_{\text{input}}}$) depends on the overlap $\gamma$ as $O(1/\gamma)$ \cite{lin2020optimal}. However, the preparation of an input state with large $\gamma$, i.e. one that scales as $O(1/\text{poly}(n))$,  is a non-trivial task in general, and indeed forms the main bottleneck for existing ground state quantum algorithms \cite{lee2023evaluating}. One heuristic approach to initial state preparation relies on the preparation of a known classical ansatz state (such as Hartree-Fock/Kohn-Sham ground-states or tensor network states) on a quantum computer, but requires the assumption that efficient translation is possible \cite{lee2023evaluating}. Other approaches based on adiabatic quantum computing and the quantum Zeno effect have also been explored \cite{lin2020optimal}. 

Our procedure provides an alternative heuristic state preparation method that is low-cost, supplying an input quantum state directly through a low-depth ansatz $U(\theta^*)$. Compared to iterative state preparation methods, its variational nature possesses milder resource requirements. In the case that the algorithm fails to accurately produce $\ket{\lambda_k}$ due to reasons such as the limited expressibility of the ansatz or the failure of the optimizer to converge to the global optima, reasonable optimization ensures that $\ket{b(\theta^*)}$ will nonetheless be improved over the initial state $\ket{0}^{\otimes n}$, by virtue of the variational principle applied to $\tilde{H}^{-1}$. The optimized state can then be taken as the input to other exact methods that require quantum states with non-trivial overlap with $\ket{\lambda_k}$ as inputs.

\section{Applications} \label{sec:applications}
We illustrate applications of the algorithm to problems arising from chemistry and many-body physics, which are natural candidates that can exploit the improved scaling in problem size offered by QLSP solvers, and the efficient representation of physical eigenstates on quantum computers. Additional numerics on the scaling of the procedure are presented in Appendix.~\ref{appendix:additional_numerics}.

\subsection{Quantum Chemistry} \label{sec:applications_chem}
In the context of ab initio methods in chemistry and materials, our algorithm allows a primary output -- a target eigenstate itself -- to be obtained, which can then be probed for useful quantities such as dipole moments, partial charges, absorption spectra and so on. This is traditionally a difficult task in large scale, even for pioneering machine learning methods \cite{westermayr2020machine}, and is important for a variety of applications as discussed in the introduction. For further examples of usage, we refer the reader to applications considered in analagous targeted classical techniques exploiting spectral transformations, such as the references \cite{wang1994solving,sleijpen1996jacobi,dorando2007targeted,baiardi2019optimization}. 

Several strategies to select the shift $\sigma$ are also known in this context, based on extrapolation from smaller problem sizes \cite{wang1994solving} or approximations such as MPS-DMRG \cite{dorando2007targeted}, systematically guessing from results of previous iterations of the algorithm \cite{dorando2007targeted,baiardi2019optimization}, and approximations of the local density of states \cite{kecceli2018siesta,williams2020shift} via classical techniques such as the Kernel Polynomial Method \cite{weisse2006kernel}.

As an illustration, Fig.~(\ref{fig:lih dissociation}) shows numerical simulations of our algorithm applied to prepare highly excited ($k = 500^{th}$) eigenstates of a 10-qubit molecular Hamiltonian -- the Lithium Hydride (LiH) molecule in the STO-3G basis -- at different bond lengths to recover its potential energy surface, compared against spectral folding, another existing variational technique. We pre-compute its spectrum exactly as comparison, and to select the shifts at different bond lengths. We consistently obtain superior results (smaller errors and standard deviations) with shift-inversion across all bond lengths considered, which is expected to further differentiate with increasing problem sizes.

In practice, further physical quantities can be extracted from the output quantum states in the form of easily accessible statistical quantities (see e.g. \cite{whitfield2011simulation} for relevant observables such as the number and excitation operators in the form of sums of Pauli terms).

\begin{figure}
\centering
\begin{subfigure}{0.47\textwidth}
\includegraphics[width=\textwidth]{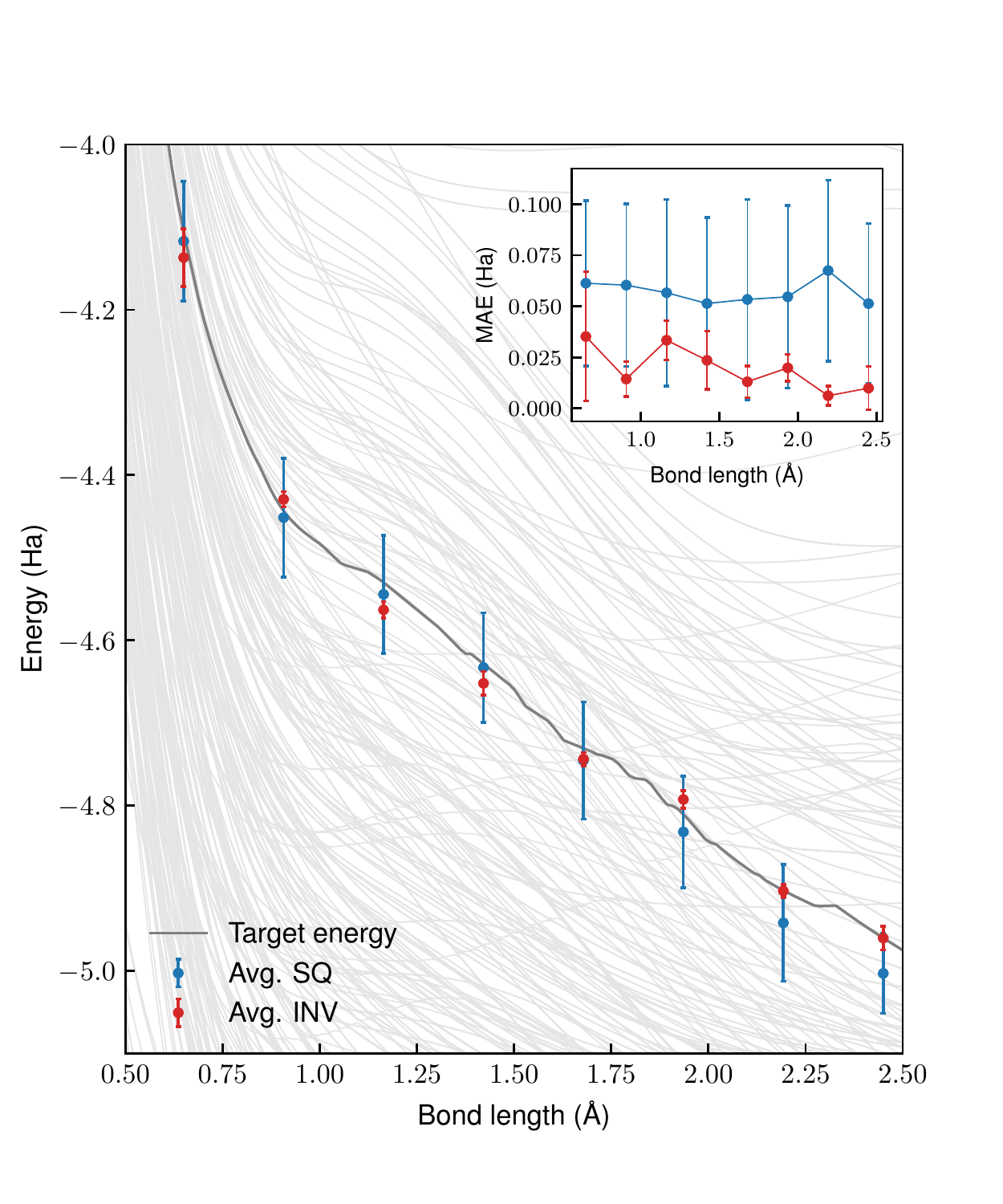}
\end{subfigure}
\caption{Potential energy surface of LiH computed with the shift-inversion algorithm (blue) and the spectral folding algorithm (red), compared to ED (dark grey line). An eigenstate near the middle of the spectrum ($k = 500$) is targeted at all bond lengths. Exact potential energy surfaces for other energies are also shown (light grey lines). Each datapoint is obtained by averaging 100 runs of the algorithm with randomly chosen initial parameters, and standard deviations are indicated by vertical error bars. The ansatz (alternating single-qubit rotations and fully-entangling CNOTs, specified in Appendix.~\ref{appendix:additional_numerics}) with 12 layers and the SLSQP/BFGS optimizers are used. Inset shows Mean Absolute Error (MAE) from exact results, with errors bars indicating standard deviations. We observe better accuracy and consistent convergences with shift-inversion.}
\label{fig:lih dissociation}
\end{figure}

\subsection{Spin systems} \label{sec:applications_spin}
Many-body quantum spin systems with local interactions present another natural application domain for our proposed algorithm. They are particularly interesting from a quantum computing viewpoint due to the difficulty of their simulation on classical computers (when approximative representations such as tensor networks are inapplicable), while remaining efficient on quantum computers \cite{childs2018toward}. They are also attractive for near-term applications of quantum computers due to milder resource requirements that result from the locality and sparsity of interactions between spins.

In particular, we consider systems exhibiting many-body localization (MBL), which have recently become a subject of interest in our context due to their remarkable information-theoretic properties \cite{bauer2014analyzing,smith2016many}. Depending on the magnitude of disorder present, their eigenstates experience a transition in entanglement between volume-law and area-law, with the former case requiring full $2^n$-dimensional wavefunctions to be specified, while the latter case permits efficient description by tensor networks. Furthermore, due to the presence of a mobility edge in energy, localization-ergodic transitions occur at different critical disorders for different energies \cite{luitz2015many,abanin2019manybody}. A detailed understanding of this dynamical phase transition therefore entails a study of the entire spectrum beyond the ground state, requiring the preparation and study of highly excited eigenstates. 

Concretely, we present numerical results on preparing excited eigenstates of the 1D disordered isotropic Heisenberg model of $L$ spins with the Hamiltonian:
\begin{equation} \label{eq:xxx_ham}
    H_{\text{MBL}} = J \sum_{i=1}^{L-1} \vec{S}_i \cdot \vec{S}_{i+1} + \sum_{i=1}^L h_i S^z_i,
\end{equation}
where $\vec{S}_i = (S^x_i, S^y_i, S^z_i)$ is the vector of spin-1/2 operators at site $i$, $J = 1$ the spin-spin interaction strength, and $h_i$ the strength of the transverse magnetic field at site $i$ which is randomly distributed in the interval $[-W,W]$, where $W$ is the disorder strength. Notably, current state-of-the-art results obtained with classical algorithms for this model across both phases are ultimately limited to  $L \sim 30$ spins due to exponentially demanding resource requirements \cite{pietracaprina2018shift,sierant2020polynomially}.

We illustrate how our algorithm can be used to extract information from the full excitation spectrum of a Hamiltonian without prior knowledge on eigenstate energies, as is typical for studies of MBL across its full spectrum \cite{luitz2015many,khemani2016obtaining,devakul2017obtaining,yu2017finding,pietracaprina2018shift}. In short, the algorithm is initialized by selecting an initial state $\ket{b(\theta_{\text{init}})}$ and taking its energy $\expt{b(\theta_{\text{init}})}{H}{b(\theta_{\text{init}})}$ to be the shift $\sigma$. When prior information on the structure of the eigenstates of $H$ are available, the intial state is chosen accordingly to maximize overlap with true eigenstates. Optimization of the shift-inverted energy then yields an eigenstate with energy closest to the initial energy. The procedure is repeated for different initial states with energies selected within a range of interest, from which further physical properties can be extracted to obtain coarse-grained information on the spectrum of $H$.

Fig.~(\ref{fig:mbl spectrum}) shows numerical results of the above procedure applied to $H_{\text{MBL}}$ for $L=12$ spins in the localized phase, when $W = 9J$ (for a particular disorder realization). Here, 10 random product states of the form $\ket{\uparrow \uparrow ... \downarrow}$ with energies distributed evenly within the spectrum of $H_{\text{MBL}}$ are chosen as initial states, as they are known to be perturbatively close to true eigenstates \cite{serbyn2013local}. The main plot shows optimization histories for the 10 chosen initial states (labelled by different colors) together with exact energies of the eigenstates of $H_{\text{MBL}}$ in the background (light grey horizontal lines), while the inset shows the probability distribution of the optimized states in the energy eigenbasis  (colors again label different initial states). The algorithm consistently converges to eigenstates of $H_{\text{MBL}}$ with high overlap with true eigenstates. Additional numerical results for Eq.~(\ref{eq:xxx_ham}) verifying the better scaling of our method compared to spectral folding are discussed in Appendix.~\ref{appendix:mbl}.

As noted in Section. \ref{sec:near_term_analysis}, the near-term implementation for this class of problem is especially interesting, due to low implementation cost (a further reduction in the number of terms by a factor of $\sim$3 for the case of Eq.~(\ref{eq:xxx_ham}) can be achieved via anti-commutative grouping), favourable scaling results from \cite{bravo2019variational} for a similar Hamiltonian, and the success of analogous classical approaches that represent solution of QLSP with MPS.

Finally, we remark on the illustrative nature of the above numerics; in practice, eigenstates in the MBL phase in 1D are well-approximated by MPS \cite{khemani2016obtaining,yu2017finding}, so it is unlikely that an exponential improvement over methods based on classical DMRG combined with spectral transformation can be found from quantum computers for this task. Nonetheless, as discussed in Section.~\ref{sec:advantage}, this does not rule out meaningful polynomial improvements \cite{babbush2021focus,lee2023evaluating}. 

\begin{figure}
\centering
\begin{subfigure}{0.5\textwidth}
\includegraphics[width=\textwidth]{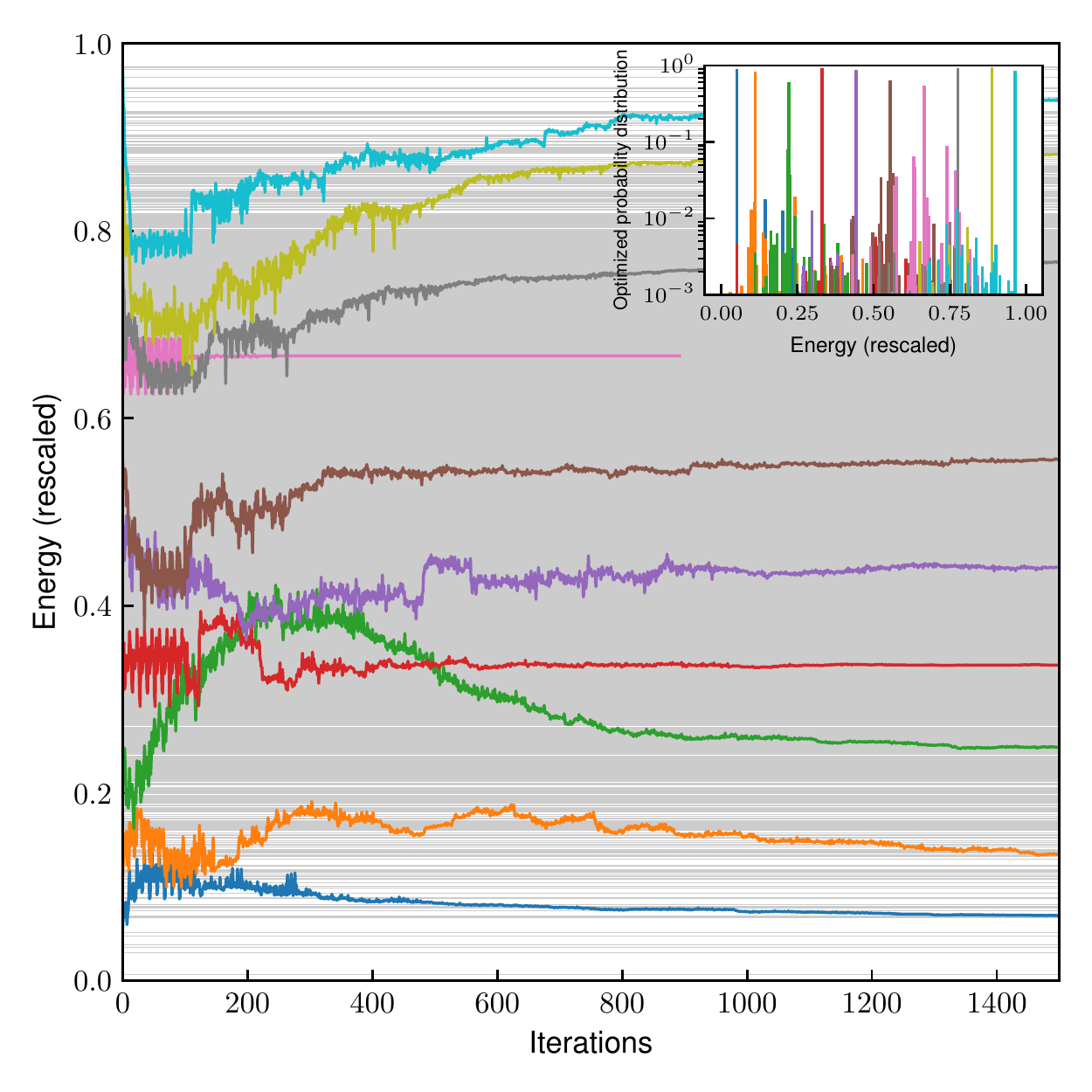}
\end{subfigure}
\caption{Extraction of eigenstates from the spectrum of the MBL Hamiltonian of Eq.~(\ref{eq:xxx_ham}) with $L=12$ spins, at $W = 9J$. 10 initial product states with energies distributed evenly in the spectrum of the Hamiltonian were chosen as initial states of the optimization procedure, with their energies taken as the shift $\sigma$. The main plot shows energy histories of different initial states (labelled by 10 colors) throughout optimization, with grey horizontal lines denoting energies of all eigenstates of $H_{\text{MBL}}$. The inset shows the probability distribution of the final optimized states (labelled by the same 10 colors) in the same eigenbasis, where the y-axis is in logarithmic scale. In the main plot, the grey region due to overlapping grey horizontal lines near the middle of the spectrum is due to the high density of states. An ansatz (alternating single-qubit rotations and fully-entangling CNOTs, specified in Appendix.~\ref{appendix:additional_numerics}) with 7 layers and the COBYLA optimizer is used.}
\label{fig:mbl spectrum}
\end{figure}




\section{Discussion and outlook} \label{sec:discussion}
This work introduces a quantum algorithm for the task of preparing highly excited eigenstates in a targeted manner, which is enabled by the ability to evaluate expectation values of inverse Hamiltonians on a quantum computer. The core advantage of the proposed algorithm lies in its natural and efficient usage of QLSP solvers, which circumvents their usual drawbacks when applied to generic linear algebraic problems. We detail and analyze implementations of the algorithm suitable for both near-term and fault-tolerant error corrected quantum computers. We also examine conditions and situations under which our algorithm outperforms existing classical and quantum algorithms, and find problems arising from physical contexts such as chemistry and many-body physics to be promising candidates. Numerical simulations on quantum chemical models and disordered many-body spin systems further supplement our arguments, and illustrate the algorithm's practical applications.

The practical utility of our algorithm paves the way for possible studies of physics beyond the ground state on quantum computers. In comparison to the well-studied ground state problem, the case for highly excited eigenstates remain relatively unexplored, in part due to the lack of scalable state preparation methods on quantum computers. Our study sets the stage for a number of follow-up investigations in this direction, including the need to further understand the representation of excited eigenstates on quantum computers (which manifests as the availability of efficient ansatz choices $U(\theta)$ in our context), and more detailed investigations on problem classes that satisfy the requirements of our algorithm, which are subjects of ongoing work.

Another main contribution of our work is in identifying a natural context for the application of QLSP solvers. Following discussions in Sections.~\ref{sec:algorithm} and \ref{sec:advantage}, this is a non-trivial task in general due to specific requirements that must be satisfied for QLSP solvers to remain efficient. Earlier attempts focused on replacing linear-algebraic subroutines present in generic settings such as machine learning \cite{wiebe2012quantum,lloyd2013quantum,rebentrost2014quantum} and the solution of differential equations \cite{berry2014high} with QLSP solvers, but the lack of sufficient structure prohibits a thorough analysis with respect to the QLSP solvers' conditions \cite{aaronson2015read}. Along with recent work in similar directions \cite{tong2021fast}, our work reveals that a remarkably natural context can be established when the inputs and output of the QLSP solver correspond to physically meaningful objects such as Hamiltonians and eigenstates of physical systems. Furthermore, while the present work leverages the ability to evaluate inverse expectation values of the form Eq.~(\ref{eq:cost_func}) to prepare highly excited eigenstates on quantum computers, we expect our analyses to be relevant beyond this task. Indeed, functionals of the same form appear in other contexts, such as in studies of the dynamics and high temperature properties of strongly interacting many-body systems via Green's function, which can directly leverage our implementations and analyses. We elaborate on these alternative uses in Appendix.~\ref{appendix:applications}.

As discussed in Section.~\ref{sec:advantage}, the variational aspect of our algorithm precludes vigorous scaling arguments, which limits our knowledge on its performance at large scales. Combined with the difficulty of simulating large quantum circuits beyond sizes accessible by exact diagonalization, performance guarantees will likely require implementations on real quantum computers to establish. On this front, hardware advancements in recent years have lead to the availability of quantum computers with qubit counts in the order of 100s (see, e.g. \cite{kim2023evidence}). Consistent with these advancements, our resource analyses indicate feasibility, especially for spin systems defined on a low dimensional lattice (e.g. of the kind discussed in Section.~\ref{sec:applications_spin}, and possibly with a lattice geometry conforming to the device's topology), where the near-term implementation with VQLS (described in Section.~\ref{sec:near_term}) has a manageable $O(nD)$ resource requirement. This opens up exciting avenues for experimentation in the near-term in classically intractable regimes, such as in studies of disordered spin systems across both ergodic-MBL phases on dimension $D>1$, where a complete understanding of this phase transition necessitates the study of highly excited eigenstates due to the presence of a mobility edge. These consideration also introduces additional practical challenges, such as the development of hardware-friendly QLSP solvers, ansatz designs for excited eigenstates, and strategies to address more resource-intensive problems (such as via efficient term grouping, encoding, and measurement \cite{tilly2022variational, choi2023measurement}).

Finally, as detailed in Section.~\ref{sec:other_components}, a limitation of our algorithm is the unfavorable dependence on the spectral gap $\Delta$, an issue common to classical methods. This complicates implementation when targeting eigenstates located at regions in the spectrum with exponentially growing density of states, which is expected to be common beyond specific cases (such as certain electronic structure problems \cite{kecceli2016shift} and Anderson localized models \cite{elsner1999anderson,pietracaprina2018shift}). As is the case for classical methods \cite{pietracaprina2018shift,baiardi2019optimization}, preconditioning and root-homing methods are expected to be crucial in practical implementations, and their adaptations for quantum algorithms are active areas of research \cite{tong2021fast}. It is also worth investigating whether information on the structure of eigenstates and the presence of symmetries can be directly encoded in the quantum algorithm itself (e.g. in the form of the cost function, or that of the ansatz) to resolve exponentially concentrating eigenstates, in a manner more natural than classical methods \cite{khemani2016obtaining,devakul2017obtaining}. Otherwise, certain situations -- such as studies of the kind described in Section.~\ref{sec:applications_spin} in MBL, where only global, relative information on the spectrum is relevant -- permit a relaxation of this dependence, by allowing for a superposition of eigenstates within a fixed energy interval, instead of targeting a specific eigenstate. 

\section{Acknowledgments}
We thank Chee-Kong Lee, Liang Shi, and Patrick Rebentrost for useful discussions. LCK and SHC thank the Ministry of Education, Singapore and the National Research Foundation Singapore for their support.

\clearpage

\appendix

\section{The HHL algorithm as a subroutine} \label{appendix:hhl}
This section briefly describes a standard implementation of the HHL algorithm and its modification in our proposed algorithm (shown in Fig.~(\ref{fig:si_hhl})), and establishes notations used throughout Sections.~\ref{sec:fault_tol} and \ref{sec:fault_tol_anal}.

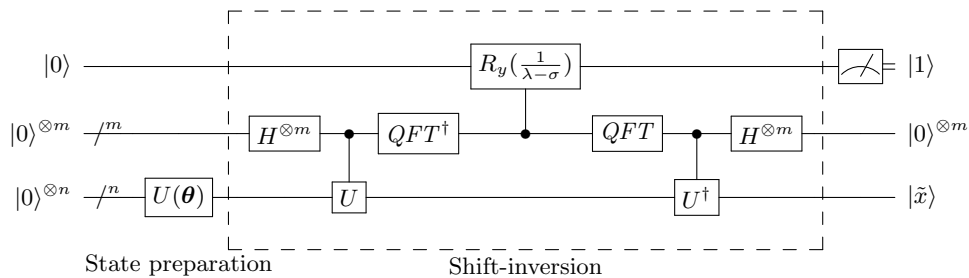
\begin{figure*}
\[
\begin{array}{c}
\Qcircuit @C=0.5em @R=1em {
      &&&&& \mbox{} &&&& \mbox{} &&&&&&&\\
      \lstick{\ket{0}}    & \qw & \qw      & \qw & \qw & \qw & \qw & \qw & \qw      & \qw               & \qw                                                & \gate{R_y(\frac{1}{\lambda - \sigma})} & \qw                                        & \qw              & \qw   & \qw    & \qw  & \meter & \rstick{\ket{1}} \cw \\
      \lstick{\ket{0}^{\otimes m}}  & \qw & {/^m}\qw & \qw & \qw & \qw & \qw & \qw & \gate{H^{\otimes m}} & \ctrl{1}          & \gate{\mathit{QFT}^\dagger} & \ctrl{-1}                & \gate{\mathit{QFT}} & \ctrl{1}         & \gate{H^{\otimes m}} & \qw  & \qw  & \qw   & \rstick{\ket{0}^{\otimes m}} \qw \\
      \lstick{\ket{0}^{\otimes n}}  & \qw & {/^n}\qw & \qw & \qw & \gate{U(\bm{\theta})} & \qw & \qw & \qw      & \gate{U} & \qw                                                & \qw                      & \qw                                        & \gate{U^\dagger} & \qw  & \qw  & \qw  & \qw    & \rstick{\ket{\tilde{x}}} \qw \\
      &&&&& \mbox{} &&&& \mbox{} &&&&&&&\\
      &&&&& \mbox{State preparation} &&&&&& \mbox{Shift-inversion} &&&&&&&\\
    \relax \gategroup{1}{8}{5}{16}{.7em}{--}
    }
\end{array}
\]
\caption{The modified HHL circuit for shift-inversion applied after a variational state preparation circuit $U(\bm{\theta})$. At the end of the optimization procedure, $\ket{b(\bm{\theta^*})} = U(\bm{\theta^*})\ket{0}^{\otimes n}$ prepares the state which minimizes the shift-inverted expectation value $\bra{b(\bm{\theta})} (H-\sigma)^{-1} \ket{b(\bm{\theta})}$, corresponding to an approximation of the eigenstate with energy close to $\sigma$. $n$ qubits form the \textit{solution register}, $m$ qubits form the \textit{eigenvalue register}, and a final qubit forms the \textit{ancillary register}.}
\label{fig:si_hhl}
\end{figure*}

For an $n$-qubit input quantum state $\ket{b}$ of dimension $N = 2^n$ and input matrix $H$ of dimension $N \times N$, sparsity $s$ and condition number $\kappa$, the HHL algorithm \cite{harrow2009quantum} solves the QLSP with a runtime that scales as $O(\log(N) s^2 \kappa^2 / \epsilon)$, which notably achieves an exponentially improved scaling in matrix dimension $N$ compared to classical LSP solvers. The algorithm's circuit consists of 3 main sets of registers : $n$ qubits form the \textit{solution register}, $m$ qubits form the \textit{eigenvalue register}, and a final qubit forms the \textit{ancillary register}. Expanding $\ket{b}$ in the eigenbasis of $H = \sum_i \lambda_i \ket{\lambda_i} \bra{\lambda_i}$ yields $\ket{b} = \sum_i \beta_i \ket{\lambda_i}$. The following sequence of operations describe the essential components of the HHL algorithm:
\begin{enumerate}
    \item \textbf{Preparation of input quantum state $\ket{b}$} : Starting from an easily preparable quantum state $\ket{0}^{\otimes n}$, the unitary operation $U$ prepares the input quantum state $\ket{0}^{\otimes n} \rightarrow \ket{b}$ in the solution register.
    
    \item \textbf{Quantum Phase Estimation (QPE)} : QPE is applied with the solution register as the target and the eigenvalue register as the control, resulting in the transformation:
    \begin{equation}
        \sum_i \beta_i \ket{0}^{\otimes m} \ket{\lambda_i} \rightarrow \sum_i \beta_i \altket{\tilde{\lambda}_i} \ket{\lambda_i},
    \end{equation}
    where $\altket{\tilde{\lambda}_i}$ encodes the eigenvalue $\lambda_i$ up to $m$ bits, denoted $\tilde{\lambda}_i$.
    
    \item \textbf{Eigenvalue inversion} : A controlled $R_y (\phi)$ rotation with the eigenvalue register as the control, and the ancillary register as the target results in the transformation:
    \begin{equation} \label{eq:eigval_inv}
    \begin{split}
       \sum_i \beta_i \altket{\tilde{\lambda}_i} \ket{\lambda_i} \ket{0} \rightarrow & \sum_i \bigg[ \beta_i \sqrt{1 -  \frac{C^2}{\tilde{\lambda}_i^2}} \altket{\tilde{\lambda}_i} \ket{\lambda_i} \ket{0} \\ 
       & + \beta_i \frac{C}{\tilde{\lambda}_i} \altket{\tilde{\lambda}_i} \ket{\lambda_i} \ket{1} \bigg],
    \end{split}
    \end{equation}
    with $\phi = \text{arccos}(C/\tilde{\lambda})$, where $C$ is a chosen constant. The circuit to encode $\phi$ onto another set of $d$ qubits (which we hide in our description here) can be executed efficiently with arithmetic circuits, which requires resources that scale only polynomially with $d$ \cite{lin2022lecture,vazquez2022enhancing}. To ensure the unitarity of this step, the choice of $C$ must also be chosen to satisfy $C \leq \tilde{\lambda}_i ~~ \forall j$. At the same time, it should be maximized (to maximize postselection probability $p_1$, due to Eq.~(\ref{eq:proba1}) later), resulting in the choice:
    \begin{equation} \label{eq:c_cond}
        C = \min_i \{ \tilde{\lambda}_i \}.
    \end{equation}
    \item \textbf{Inverse QPE} : With the solution register as the target and the eigenvalue register as the control, inverse QPE un-computes the eigenvalue register, allowing it to be discarded:
    \begin{equation}
        \begin{split}
       \sum_i \left[\beta_i \sqrt{1 -  \frac{C^2}{\tilde{\lambda}_i^2}} \altket{\tilde{\lambda}_i} \ket{\lambda_i} \ket{0} + \beta_i \frac{C}{\tilde{\lambda}_i} \altket{\tilde{\lambda}_i} \ket{\lambda_i} \ket{1} \right]
       \rightarrow \\ 
       \sum_i \left[ \beta_i \sqrt{1 -  \frac{C^2}{\tilde{\lambda}_i^2}} \ket{\lambda_i}  \ket{0} + \beta_i \frac{C}{\tilde{\lambda}_i} \ket{\lambda_i} \ket{1} \right].
        \end{split}
    \end{equation}
    
    \item \textbf{Measurement of ancilla qubit and post-selection } : Finally, measurement of the ancilla qubit, and post-selecting only states with outcome `1' results in the quantum state in the solution register to be:
    \begin{equation}
        \sum_i \frac{\beta_i}{\tilde{\lambda}_i} \ket{\lambda_i} \equiv \ket{\tilde{x}} \approx \ket{x},
    \end{equation}
    where $\ket{\tilde{x}}$ corresponds to an approximation to $\ket{x} = A^{-1} \ket{b} = \sum_i \frac{\beta_i}{\lambda_i} \ket{\lambda_i}$, the solution of the QLSP. The probability $p_1$ of successfully observing the outcome `1' is:
    \begin{equation} \label{eq:proba1}
        p_1 = \left\Vert \sum_i C \frac{\beta_i}{\tilde{\lambda}_i} \ket{\lambda_i} \right\Vert^2 \approx C^2 \| A^{-1} \ket{b} \|^2,
    \end{equation}
    which depends on both $C$ and the norm of $A^{-1} \ket{b}$.

\end{enumerate}

Within our algorithm, the eigenvalue inversion part of the HHL circuit is modified to include a shift $\sigma$, while the input quantum state $\ket{b(\theta)}$ is prepared by a unitary ansatz $U(\theta)$, resulting in the circuit shown in Fig.~(\ref{fig:si_hhl}). This changes the choice of the constant $C$ from Eq.~(\ref{eq:c_cond}) to:
\begin{equation} \label{eq:c_cond_sigma}
        C = \min_i \{ |\tilde{\lambda}_i - \sigma| \},
    \end{equation}
which introduces a dependence on the spectral gap $\Delta$ in $p_1$, due to Eq.~(\ref{eq:proba1}). 

Upon successfully post-selecting, the observable $\tilde{H}$ is measured on the output quantum state $\ket{\tilde{x}}$ to compute the cost function $C(\theta) = \bra{b(\theta)} \tilde{H}^{-1} \ket{b(\theta)}$, and a classical optimizer selects the next set of parameters $\theta'$ in an iterative manner to search for the optimal set of parameters $\theta^*$ that minimizes $C(\theta)$. This procedure results in a quantum state $\ket{b(\theta^*)}$ that best approximates the target eigenstate $\ket{\lambda_k}$ within the set of states reachable by the ansatz $U(\theta)$. As a corollary, $p_1$ increases as optimization progresses, tending to unity as $\beta_k$ approaches 1 (from Eq.~(\ref{eq:proba1}), if the $\beta_i's$ concentrate at $\beta_k$, $C$ approximately cancels with the only remaining denominator $\tilde{\lambda}_k - \sigma$).

\section{Additional numerical results} \label{appendix:additional_numerics}
In this section, we present additional numerical simulations on the algorithm's scaling in performance compared to the shift-squared spectral transform, and typical cost landscapes for both the fault-tolerant and near-term implementation (via HHL and the VQLS respectively).

Unless specified, classical simulations of quantum circuits in our work were performed exactly using the Qiskit \cite{aleksandrowicz2019qiskit} software package. For the HHL algorithm, Hamiltonian simulation within QPE and inverse QPE was performed exactly (which suppresses errors in implementing $e^{iHt}$), while eigenvalue inversion was performed through uniformly controlled rotations \cite{mottonen2004transformation}. The form of the layered variational ansatzes (for both $U(\theta)$ and $V(\phi)$) used in this section and the main text is displayed in Fig.~(\ref{fig:numerical_ansatz}), and consists of alternating layers of single qubit rotations and fully-entangling CNOTs. The default choice of classical optimizer is either SLSQP or BFGS.

\begin{figure*}
\[
\begin{array}{c}
\Qcircuit @C=1.0em @R=0.2em @!R { \\
	 	\nghost{{q}_{0} :  } & \lstick{{q}_{0} :  } & \gate{\mathrm{RX}\,(\mathrm{{\ensuremath{\theta_0}}})} \barrier[0em]{2} & \qw & \ctrl{1} & \ctrl{2} & \qw \barrier[0em]{2} & \qw & \gate{\mathrm{RX}\,(\mathrm{{\ensuremath{\theta_3}}})} \barrier[0em]{2} & \qw & \ctrl{1} & \ctrl{2} & \qw \barrier[0em]{2} & \qw & \gate{\mathrm{RX}\,(\mathrm{{\ensuremath{\theta_6}}})} & \qw & \qw\\
	 	\nghost{{q}_{1} :  } & \lstick{{q}_{1} :  } & \gate{\mathrm{RX}\,(\mathrm{{\ensuremath{\theta_1}}})} & \qw & \control\qw & \qw & \ctrl{1} & \qw & \gate{\mathrm{RX}\,(\mathrm{{\ensuremath{\theta_4}}})} & \qw & \control\qw & \qw & \ctrl{1} & \qw & \gate{\mathrm{RX}\,(\mathrm{{\ensuremath{\theta_7}}})} & \qw & \qw\\
	 	\nghost{{q}_{2} :  } & \lstick{{q}_{2} :  } & \gate{\mathrm{RX}\,(\mathrm{{\ensuremath{\theta_2}}})} & \qw & \qw & \control\qw & \control\qw & \qw & \gate{\mathrm{RX}\,(\mathrm{{\ensuremath{\theta_5}}})} & \qw & \qw & \control\qw & \control\qw & \qw & \gate{\mathrm{RX}\,(\mathrm{{\ensuremath{\theta_8}}})} & \qw & \qw\\
\\ }
\end{array}
\]
\caption{The layered variational ansatz used in all numerical simulations, shown in the case of $n=3$ qubits and 2 layers/repetitions. The circuit begins with a set of individually parametrized single-qubit RX rotation gates. Subsequently for each layer, individually parametrized single-qubit RX rotation gates are followed by CZ gates applied among all qubits. The number of parameters scale as $n(1+n_{layers})$.}
\label{fig:numerical_ansatz}
\end{figure*}
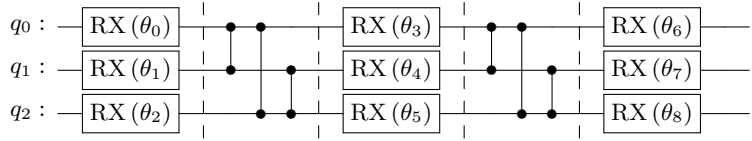

\subsection{Performance compared to shift-squared spectral transform} \label{appendix:scaling}
This section investigates the performance of our algorithm with increasing problem sizes compared to an existing, conceptually similar targeted algorithm. In particular, we compare the performance of the optimization procedure present in our algorithm to spectral folding/shift-squared spectral transformation. Replacing the shift-inverted expectation value of our algorithm ($\bra{b(\theta)} \tilde{H}^{-1} \ket{b(\theta)}$, Step. 3) with the shift-squared expectation value $\bra{b(\theta)} \tilde{H}^2 \ket{b(\theta)}$ results in the spectral folding/shift-squared spectral transform, which has been proposed to prepare excited eigenstates on near-term quantum computers \cite{peruzzo2014variational,santagati2018witnessing,zhang2021adaptive}. It can be evaluated straightforwardly on a quantum computer either from the second moment of the observable $\tilde{H}$, or by expanding and evaluating the terms of $\tilde{H}^2$ separately, both having runtimes quadratically more expensive than estimating $\langle H \rangle$. As mentioned in the main text, while computationally less demanding than shift-inversion, it faces difficulties when targeting eigenstates located at regions of high local density of states (LDOS) as the quadratic spectral transformation further shrinks energy gaps. On the contrary, shift-inversion either repels nearby eigenstates or transforms them to the other end of the spectrum, and conventionally outperforms shift-squaring in classical contexts \cite{pietracaprina2018shift} (c.f. Section.~\ref{sec:relation_quantum}).

In the following subsections, we compare the two algorithms' performances when targeting eigenstates located at positions with increasing LDOS, and for increasing system sizes. We consider diagonal matrices with LDOS that concentrate at the middle of the spectrum, and spin Hamiltonians arising from the disordered isotropic Heisenberg model.

\subsubsection{Scaling with increasing DOS}
To investigate the performance of the algorithm with increasing LDOS, we first consider matrices with LDOS that concentrate in the middle of the spectrum as inputs to our algorithm, and attempt to solve for eigenstates located at different energy densities. In particular, we consider matrices with eigenvalues in $(0,1)$ distributed either according to a Gaussian ($p_G(x) = \frac{1}{s \sqrt{2\pi}}e^{-\frac{1}{2}(\frac{x-\mu}{s})^2}$) or Laplace ($p_L(x) = \frac{1}{2b}e^{-\frac{|x-\mu|}{b}}$) distribution, with parameters $\mu = 0.5$, $s = 0.1$ and $b = 1$. To target an eigenstate with energy $\lambda_k$, the choice of the shift is always $\sigma = \lambda_k + 0.25\Delta$.

Two simplifications were adopted in our numerics in this subsection. Firstly, input matrices were taken to be diagonal in the computational basis, so that their eigenstates are the computational basis states, which are easily preparable product states. This guarantees that target eigenstates are always located within the set of reachable states of our ansatz (Fig.~(\ref{fig:numerical_ansatz})). The spectrum of the input matrix can also be specified by distributing the values of its diagonal entries in an arbitrary manner. Secondly, both cost functions were computed exactly via matrix multiplication, which removes errors from the QLSP solver and expectation estimation. Taken together, these simplifications remove the issue of reachability and imprecisions due to solving the QLSP and expectation estimation, so that the results purely reflect how different choices of cost functions impact the algorithm's final accuracy.

\begin{figure*}
\centering
\begin{subfigure}{0.47\textwidth}
    \includegraphics[width=\textwidth]{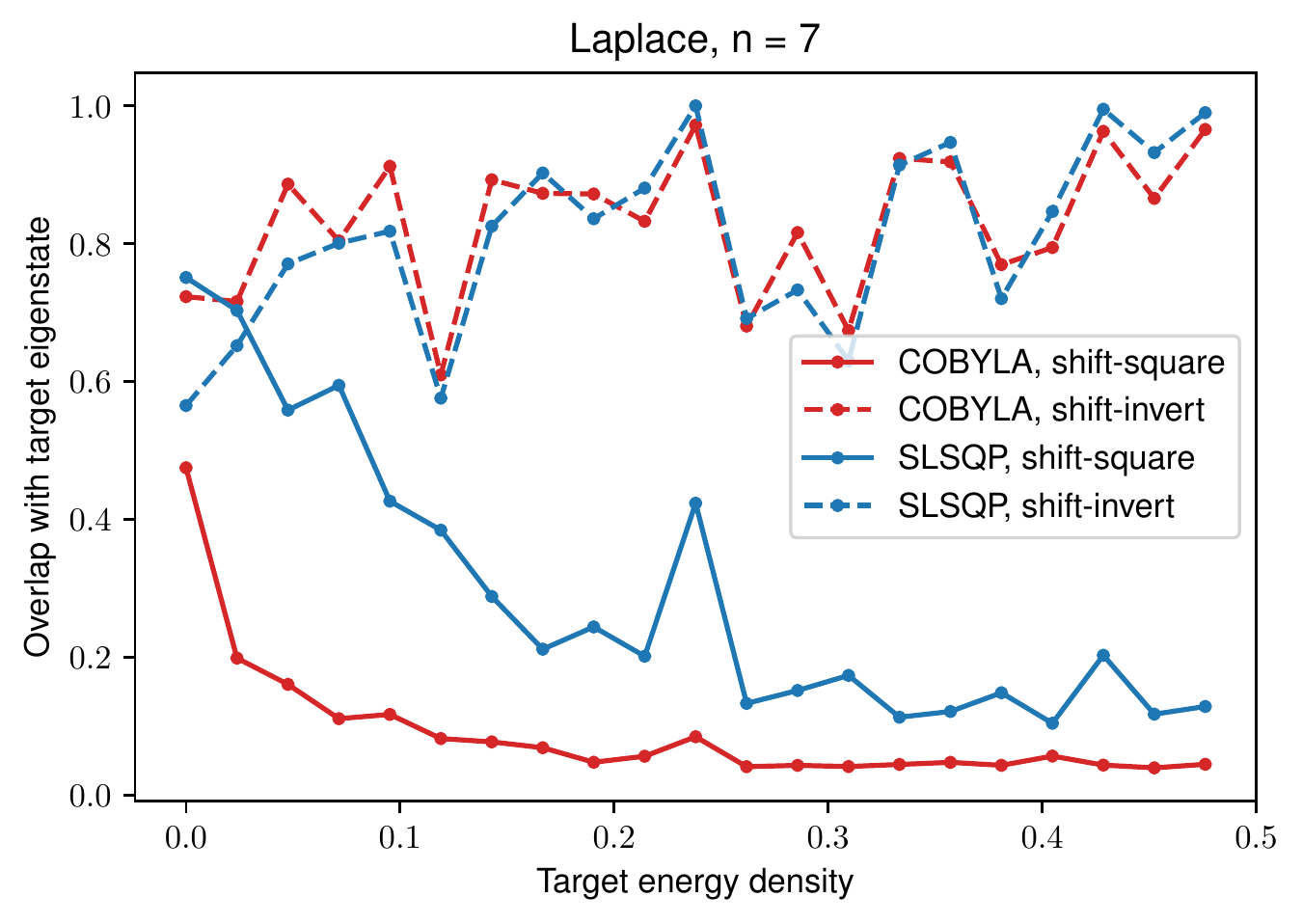}
\end{subfigure}
\begin{subfigure}{0.47\textwidth}
    \includegraphics[width=\textwidth]{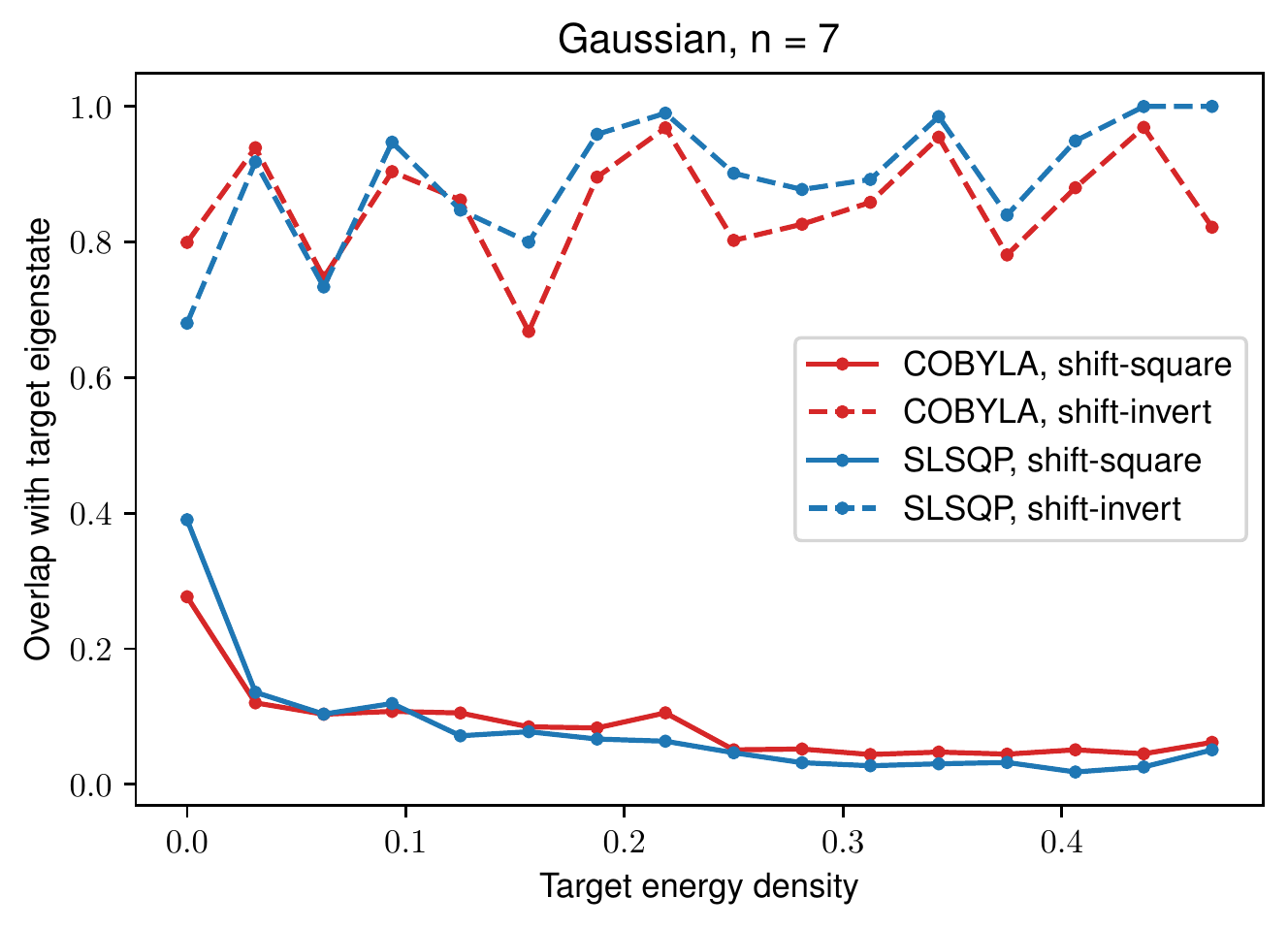}
\end{subfigure}
\caption{Plots displaying the performance of shift-inversion (dotted lines) versus shift-squaring (solid lines) for different Hamiltonians (as title of subplot) and optimisers (as different colored lines). Each datapoint is averaged over 150 optimization runs with random initial parameters to quantify the average performance of the algorithm. The y-axis shows the overlap with the target eigenstate (higher is better), while the x-axis corresponds to the energy density of target eigenstates : the smallest value 0 corresponds to the ground state, while the largest value 0.5 corresponds to the state at the middle of the spectrum ($2^7/2 = 64$-th eigenstate), where the LDOS is maximum. Much higher overlap is observed in most cases with shift-inversion compared to shift-square, especially towards the middle of the spectrum, as the LDOS tends to its maximum.}
\label{fig:scaling}
\end{figure*}

Simulation results for $2^7$-dimensional matrices are displayed in Fig.~(\ref{fig:scaling}). In most cases, for both types of matrices, much higher overlap is observed with shift-inversion compared to shift-squaring. This is especially pronounced towards the middle of the spectrum (where the LDOS tends to its maximum), where shift-inversion consistently succeeds with near-unity overlap, while spectral folding fails with near-zero overlap. These results confirm our intuition on their performances, and verifies the trainability of the shift-inverted cost function.

\subsubsection{Scaling with size for MBL Hamiltonian} \label{appendix:mbl}
Next, we compare the performance of the algorithms when applied to lattice models of increasing sizes, here chosen to be the 1D disordered isotropic Heisenberg model of $L$ spins with the Hamiltonian:
\begin{equation} \label{eq:appendix_xxx_ham}
    H_{\text{MBL}} = J \sum_{i=1}^{L-1} \vec{S}_i \cdot \vec{S}_{i+1} + \sum_{i=1}^L h_i S^z_i,
\end{equation}
where $\vec{S}_i = (S^x_i, S^y_i, S^z_i)$ is the vector of spin-1/2 operators at site $i$, $J = 1$ the spin-spin interaction strength, and $h_i$ the strength of the transverse magnetic field at site $i$ which is randomly distributed in the interval $[-W,W]$, where $W$ is the disorder strength. Periodic boundary conditions are adopted. Sufficiently large values of $W$ drives the model into many-body localization, in which case its density of states can be approximated as a Gaussian concentrating at the middle of the spectrum. As mentioned in the main text, the slowly growing number of terms and parallelization opportunities due to the commutation of terms renders our algorithm implementable on near-term quantum computers via the VQLS.

In our numerics, we choose $J = 1$ and $W = 8$ corresponding to the MBL phase, and target the middle ($k = 2^{L}/2$-th) eigenstate for increasing $L$, using the overlap of the optimized state with the target eigenstate as a performance metric. Again, we use the ansatz of Fig.~(\ref{fig:numerical_ansatz}) with 4 layers, and take the average over 100 disorder realizations and randomized initial parameters at each $L$. The results are shown in Fig.~(\ref{fig:mbl}); we observe consistently higher overlaps with shift-inversion. Notably, at larger problem sizes, shift-squaring always fails, preparing an eigenstate with low overlap with the target, while shift-inversion consistently obtains overlaps close to unity on at least one of the realizations (i.e. the maximum of the overlaps at each $L$ is always $< 1$ at larger sizes for shift-squaring while it is always close to $1$ for shift-inversion). 

\begin{figure}
\centering
\begin{subfigure}{0.5\textwidth}
    \includegraphics[width=\textwidth]{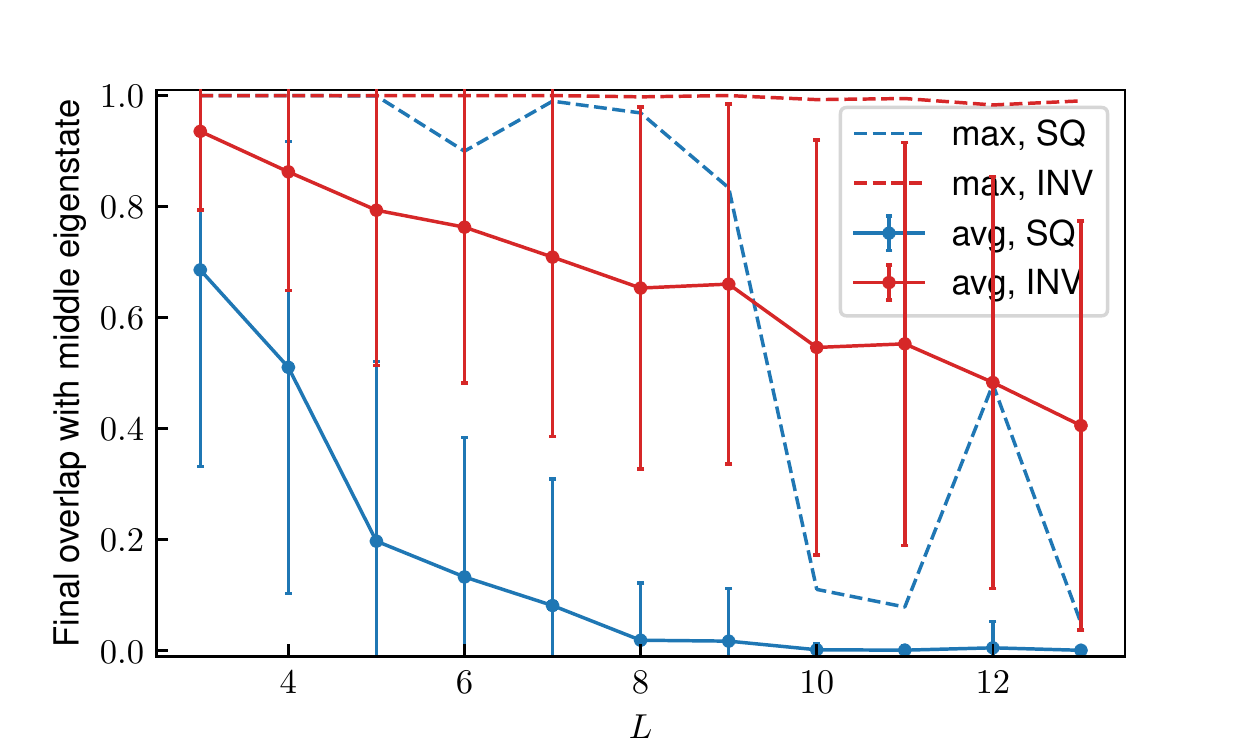}
\end{subfigure}
\caption{The overlap of the final optimized state as a function of system size $L$, for both the shift-invert algorithm (red) and the shift-squaring (blue). Each point is averaged over 100 disorder realizations and randomized initial parameters, with solid lines indicating the average overlap and dotted lines indicating the maximum. Vertical bars show the standard deviations. The maximum of the overlaps is always $< 1$ at larger sizes for shift-squaring while it is always close to $1$ for shift-inversion.}
\label{fig:mbl}
\end{figure}

\subsection{Cost/optimization landscapes}
The imprecision of the QLSP solver in Step. 3 leads to errors in the value of cost function $C(\theta) = \bra{b(\theta)} (H-s\mathbb{1})^{-1} \ket{b(\theta)}$ computed during the optimization procedure. To reduce this error, one can either increase the number of qubits in the eigenvalue register $m$ for the fault-tolerant implementation, or increase the expressivity of the ansatz $V(\phi)$ (e.g. by increasing the number of layers, and hence the number of variational parameters for a layered ansatz) for the near-term implementation.

We can observe the increase in precision by plotting the cost landscape (i.e. the value of $C(\theta)$ for a range of $\theta$) for different values of $m$ or layers. This is displayed in Fig.~(\ref{fig:landscape}) for a small $4$-qubit problem (corresponding to the electronic Hamiltonian of molecular hydrogen H$_2$), where we vary the value of a parameter (in our simulations, taken arbitrarily to be the second parameter $\theta_1$ as according to Fig.~(\ref{fig:numerical_ansatz})) while fixing that of all other parameters. As expected, cost landscapes with increasing precision (solid colored lines) (by increasing values of $m$ for HHL, or number of layers for VQLS) approach the exact value of $\bra{b(\theta)} (H-\sigma\mathbb{1})^{-1} \ket{b(\theta)}$ for when the solution of the QLSP is exact (black dotted line).

Another notable observation is that the features of the cost landscape (positions of local optima, periodicity of the landscape) are generally preserved, even at low precisions where the values of the cost function do not completely match the exact value. This is a desirable property, indicating that the optimization can tolerate the imprecise solutions of QLSPs during inverse expectation estimation, which provides a degree of robustness against errors from this source. 

\begin{figure*}
\centering
\begin{subfigure}{0.45\textwidth}
    \includegraphics[width=\textwidth]{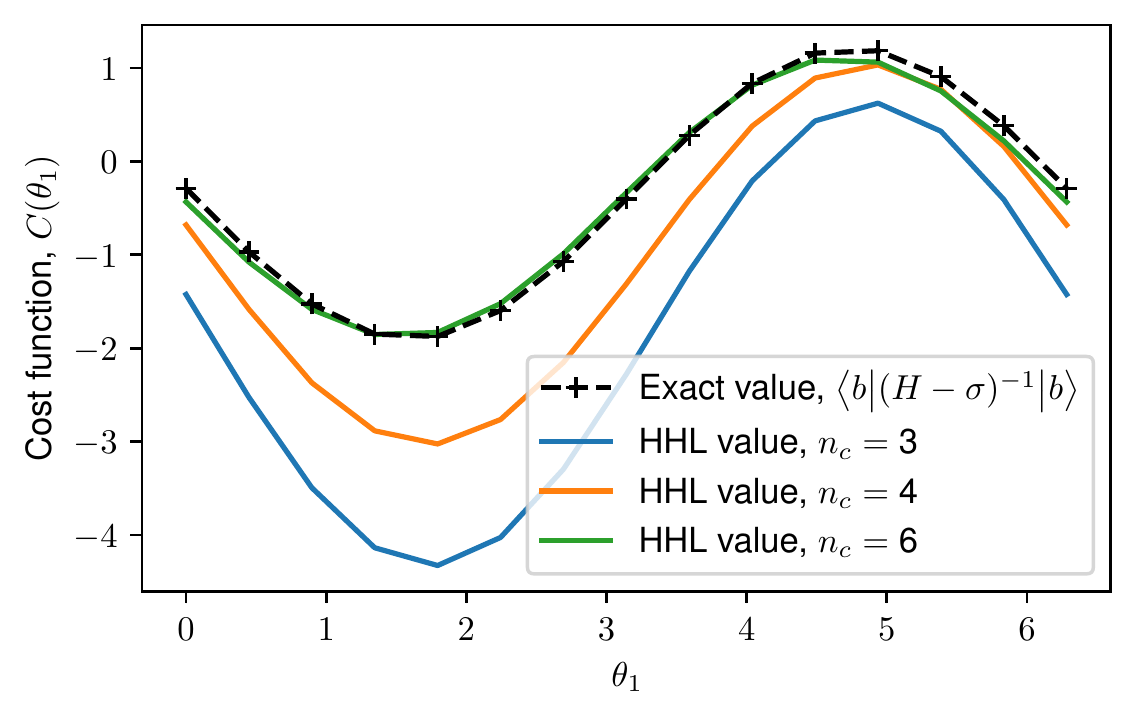}
\end{subfigure}
\begin{subfigure}{0.45\textwidth}
    \includegraphics[width=\textwidth]{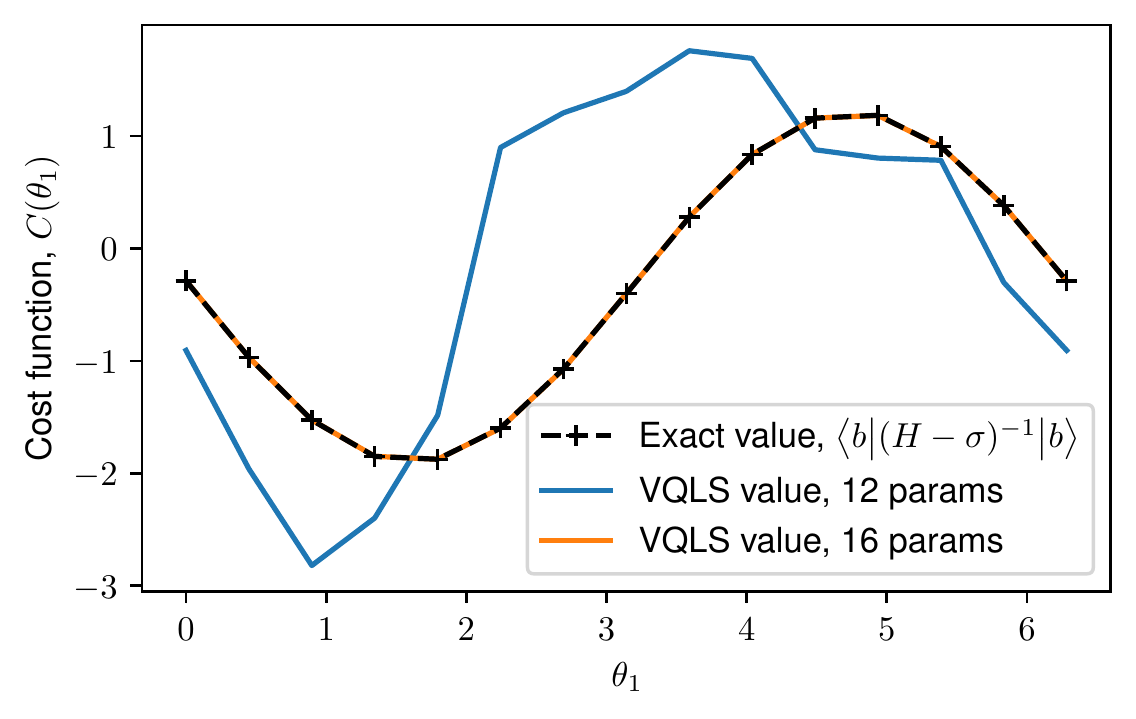}
\end{subfigure}
\caption{Plots displaying typical cost landscapes with increasing precision for HHL (left) or ansatz expressivity for VQLS (right). For VQLS, to ensure that the global optima is obtained for each QLSP solution, the best result from multiple optimizations with randomly initialized parameters is taken for each point of the landscape. Cost landscapes with increasing values of $m$ (for HHL) or number of layers (for VQLS) approach the exact value of $\bra{b(\theta)} (H-\sigma\mathbb{1})^{-1} \ket{b(\theta)}$ for when the solution of the QLSP is exact. We also observe that even at low values of $m$, the positions of the optima are generally preserved. Optimization over the imprecise landscapes will therefore still lead to correct optimal parameters $\theta^*$.}
\label{fig:landscape}
\end{figure*}

\section{Other applications} \label{appendix:applications}
This sections briefly describes examples in other areas in physics and chemistry that can benefit from the ability to efficiently prepare highly excited eigenstates of large systems, and the evaluation of the expectation value of the inverse of a Hermitian operator. 

Information beyond the ground state is generally required to access interesting physical properties (high temperature physics, dynamics, density of states, ...) of strongly correlated quantum systems. Transition amplitudes of an operator $M$ which takes the form $\expt{\lambda_i}{M}{\lambda_j}$ often appear in this context, and are challenging to evaluate on classical computers, which conventionally relies on exact diagonalization. On quantum computers, the authors of \cite{nakanishi2019subspace} consider the use of the variational SSVQE algorithm to prepare the excited eigenstates $\ket{\lambda_i}$ on a quantum computer. Similarly, \cite{endo2020calculation} considers using the same algorithm to evaluate the same quantity (appearing in Green's function through the spectral function in the Lehmann representation). In both cases, the preparation of excited eigenstates can be performed with our algorithm, especially in cases where highly excited eigenstates are involved, which becomes intractable for iterative algorithms such as SSVQE. Similarly, the proposed method to compute expectation values of inverses via Eq.~(\ref{eq:expt_relation}) can be applied to evaluate functionals of the form:
\begin{equation}
    G_{A,B}(\ket{\psi} ,z) = \expt{\psi}{O_A^{\dagger} (z-H)^{-1} O_B}{\psi},
\end{equation}
which appear in the computation of Green's function.

Conventional methods based on the variational principle are unable to solve for eigenstates corresponding to bound states of the Dirac equation (i.e. excited states of the Dirac Hamiltonian that are separated from low-energy states in the Dirac sea and high-energy states in the Fermi sea), since the spectrum of the Dirac Hamiltonian is unbounded both from above and below, a problem known as the variational collapse problem \cite{dolbeault_variational_nodate}. This problem can be overcome with classical shift-invert diagonalization techniques \cite{hill_solution_1994, fillion-gourdeau_algorithm_2017,hagino_iterative_2010}, where shift-inversion maps the bound states of the Dirac Hamiltonian to the ground state of a shift-inverted Hamiltonian. Our algorithm can thus be straightforwardly adapted for this purpose to prepare bound states of the Dirac equation on quantum computers. The same argument holds for a similar problem in models of interacting nucleons via Hartree-Fock-Bogoliubov theory \cite{tanimura2013application}.

\bibliography{sample}

\end{document}